\documentclass[11pt,reqno]{amsart}
\usepackage{amssymb,mathrsfs,graphicx}
\usepackage{ifthen}
\usepackage{hyperref}
\usepackage[margin=1in]{geometry}
\usepackage{caption}
\usepackage{sidecap}
\usepackage{rotating,dsfont,verbatim}

\usepackage{colortbl}
\definecolor{black}{rgb}{0.0, 0.0, 0.0}
\definecolor{red}{rgb}{1.0, 0.5, 0.5}
\provideboolean{shownotes} % define a boolean variable
\setboolean{shownotes}{true} % defined as ``true'' or ``false''
\newcommand{\margnote}[1]{
\ifthenelse{\boolean{shownotes}}%
{\marginpar{\raggedright\tiny\texttt{#1}}}%
{}%
}
\newcommand{\hole}[1]{
\ifthenelse{\boolean{shownotes}}%
{\begin{center} \fbox{ \rule {.25cm}{0cm} \rule[-.1cm]{0cm}{.4cm}
\parbox{.85\textwidth}{\begin{center} \texttt{#1}\end{center}} \rule
{.25cm}{0cm}}\end{center}} {} }

\title[Mean-field limit of VPFP]{On the mean-field limit of Vlasov-Poisson-Fokker-Planck equations}

\author[Chen]{Li Chen}
\address[Li Chen]{\newline Institut f\"ur Mathematik \newline
Universit\"at Mannheim, B6, 68159 Mannheim, Germany}
\email{chen@math.uni-mannheim.de}

\author[Jung]{Jinwook Jung}
\address[Jinwook Jung]{\newline Department of Mathematics and Research Institute for Natural Sciecnes \newline
Hanyang University, 222 Wangsimni-ro, 04763 Seoul, Republic of Korea}
\email{jinwookjung@hanyang.ac.kr}

\author[Pickl]{Peter Pickl}
\address[Peter Pickl]{\newline Fachbereich Mathematik \newline
Universit\"at T\"ubingen,  Auf der Morgenstelle 10, 72076 T\"ubingen, Germany}
\email{p.pickl@uni-tuebingen.de}

\author[Wang]{Zhenfu Wang}
\address[Zhenfu Wang]{\newline Beijing International Center for Mathematical Research \newline
Peking University, Beijing 100871, China}
\email{zwang@bicmr.pku.edu.cn}

\numberwithin{equation}{section}

\newtheorem{theorem}{Theorem}[section]
\newtheorem{lemma}{Lemma}[section]

\newtheorem{proposition}{Proposition}[section]
\newtheorem{remark}{Remark}[section]

\newcommand{\R}{\mathbb R}

\newcommand{\bq}{\begin{equation}}
\newcommand{\eq}{\end{equation}}
\newcommand{\e}{\varepsilon}
\newcommand{\lt}{\left}
\newcommand{\rt}{\right}

\newcommand{\pa}{\partial}

\newcommand{\ud}{\mathrm{d}}

\newcommand{\intr}{\int_{\R^3}}

\newcommand{\intrr}{\iint_{\R^3 \times \R^3}}

\makeatletter
\def\moverlay{\mathpalette\mov@rlay}
\def\mov@rlay#1#2{\leavevmode\vtop{%
   \baselineskip\z@skip \lineskiplimit-\maxdimen
   \ialign{\hfil$\m@th#1##$\hfil\cr#2\crcr}}}
\newcommand{\charfusion}[3][\mathord]{
    #1{\ifx#1\mathop\vphantom{#2}\fi
        \mathpalette\mov@rlay{#2\cr#3}
      }
    \ifx#1\mathop\expandafter\displaylimits\fi}
\makeatother

\begin{document}
%%%%%%%%%%%%%%%%
\allowdisplaybreaks

\date{\today}
\thanks{\textbf{Acknowledgment.}  This project was partially supported by the National Key R\&D Program of China, Project Number 2024YFA1015500. L. Chen is partially supported by the German Research Foundation (No. CH 955/8-1). P. Pickl is partially supported by
the Deutsche Forschungsgemeinschaft (DFG, German Research Foundation) – TRR
352 – Project-ID 470903074. Z. Wang is partially supported by NSFC grant No.12171009. 
}

%\subjclass{2010 MSC: 	35Q35,	35Q92, 76T10} 

\begin{abstract} 	The derivation of effective descriptions for interacting many-body systems is an important branch of applied mathematics. We prove a
	 propagation of chaos result  for a system of $N$ particles subject to Newtonian time evolution with or without additional white noise influencing the velocities of the particles. We assume that the particles interact according to a regularized Coulomb-interaction with a regularization parameter that vanishes in the $N\to\infty$ limit.
	 The respective effective description is the so called  Vlasov-Poisson-Fokker-Planck (VPFP), respectively the  Vlasov-Poisson (VP) equation in the case of no or sub-dominant white noise.
	 To obtain our result we combine the relative entropy method from \cite{jabinWang2016} with the control on the difference between the trajectories of the true and the effective description provided in \cite{HLP20} for the VPFP case respectively  in \cite{LP} for the VP case.
	 This allows us to prove strong  convergence of the marginals, i.e. convergence in
 $L^1$.
\end{abstract}

\maketitle \centerline{\date}

%\tableofcontents

%%%%%%%%%%%%%%%%%%%%%%%%%%%%%%%%%%%%%%%%%%%%%%%%%%%%%%%%%%%%%%%%%%%%%%%%%%%%%%%%%5
%
%
%                        Section: Introduction 
%
%
%%%%%%%%%%%%%%%%%%%%%%%%%%%%%%%%%%%%%%%%%%%%%%%%%%%%%%%%%%%%%%%%%%%%%%%%%%%%%%%%%
\section{Introduction}\label{sec:intro}
\setcounter{equation}{0}

The analytic or numerical treatment of systems composed of  many interacting agents is often difficult up to impossible. However, there are cases where simplified effective equations can be used to draw the main features of the system. This is in particular the case when the initial configurations of the agents are independent and/or  white noise is included in the dynamics. 
Assuming that propagation of chaos holds, i.e. that the independence is only mildly violated as time propagates, the law of large numbers can be used to replace the pair-interaction by its expectation value, the so-called mean-field interaction. One arrives at the effective description of the dynamics of the system.
The main ingredient for proving the validity of the effective equation is thus proving that the independence of the particles (``chaos'') in fact persists under time evolution (``propagation of chaos''). 	See for instance \cite{Sznitman91,HaurayMischler14}. 
	It is therefore not surprising that verifying the propagation of chaos for interacting many particle systems has been a vivid area of mathematical research in recent decades. Many of the underlying physical or biological systems naturally consider a singular interaction between the particles. However, depending on the system at hand, singularities in the interaction are a major obstacle to be overcome when proving the validity of propagation of chaos and its consequence: the accuracy of the effective equations.

In the present paper we will review some recent findings in this direction, where - in a sense to be described below - mildly singular interactions are considered to prove closeness of the $N$-body trajectory of the interacting system to the respective mean-field trajectory  in a probabilistic sense with respect to the initial distribution of the particles.   By ``mean-field trajectory'' we denote the trajectory on $N$-body phase space where each particle independently moves according to the one-particle flow of the effective evolution equation one obtains via the mean-field consideration explained above.  

The central insight of this paper is that such an estimate on the level of  trajectories can be used to control the relative entropy on the level of probability measures describing the full interacting system compared to the effective descriptions: one can show that for any $k\in\mathbb{N}$ the $k^{\text{th}}$ marginal $f_t^{N,k}$ of the $N$-body density $f_t^{N}$ will be close to the $k$-fold product of the respective solution of the effective equation $f_t$ in the sense that relative entropy comparing the two will be small.
Such a control of the relative entropy opens the door to understanding closeness of $f_t^{N,k}$ and $f_t^{\otimes k}$ also with respect to other prominent notions of distance, for example in the $L^1$ sense or with respect to the Wasserstein distance.

Although the findings of this paper do also apply to other models, we here focus   on  results for second order systems that are subject to white noise.
We already mentioned that the key difficulty in the derivation of mean field limits comes from the singularity of the force that one encounters for example for systems subject to Newtonian gravitation or Coulomb repulsion. To manage this difficulty  we consider a second order stochastic system with a regularized interaction force $k^N$ that will be specified below. Note already that the regularization will depend on $N$. The white noise will be modeled by a Brownian motion with an  $N$-dependent strength parameter $\sigma_N\geq 0$. We assume that, as $N\to\infty$, $\sigma_N$ converges and denote the limit by $\sigma$. Summarizing, we consider the dynamics of a (stochastic) Newtonian system with 2-body interaction $k^N$ given by 
\bq\label{micro_reg}
\begin{aligned}
	dx_t^{i,N,\sigma_N} &= v_t^{i,N,\sigma_N}dt,\\
	dv_t^{i,N,\sigma_N}&= \frac{1}{N-1}\sum_{j\neq i} k^N(x_t^{i,N,\sigma_N}-x_t^{j,N,\sigma_N})dt + \sqrt{2\sigma_N}dB_t^i.
\end{aligned}
\eq

These equations define a trajectory $(X_t, V_t)$ on $N$-body phase space and a respective flow $\Phi_t^N:\mathbb{R}^{dN}\to\mathbb{R}^{dN}$:
$$
\Phi_t^N (X_0,V_0) = (X_t, V_t) := ( x_t^{1,N,\sigma_N}, \dots,  x_t^{N,N,\sigma_N},  v_t^{1,N,\sigma_N}, \dots,  v_t^{N,N,\sigma_N}).
$$
The initial data of \eqref{micro_reg} is given by a set of independent and identically distributed (i.i.d.) square integrable random variables, i.e. \bq\label{particleID}
(x^{i,N,\sigma_N}_0,v^{i,N,\sigma_N}_0),\quad i=1,\cdots,N,
\eq
with $f_0$ be their common probability density function.
In other words, $(X_0,V_0)$ is distributed according to the density $f_0^{\otimes N}$. We denote the respective time dependent probability density on the $N$-particle phase space by $f_t^N:=f_0^{\otimes N}\circ \Phi_{-t}^N$.

From this one can guess the effective description of the system, as explained above by replacing the interaction by its expectation value, is the so-called mean-field interaction.
On the level of densities one arrives at  the following kinetic equation
\bq\label{main_eq}
\pa_t f_t + v \cdot \nabla_x f_t +(k \star \rho_t) \cdot \nabla_v f_t = \sigma \Delta_v f_t \qquad \mbox{ where }k(x) :=\pm C_d  \frac{x}{|x|^d}, \quad \sigma\geq 0,
\eq
and $\rho_t=\int_{\R^{d}}f_t \mathrm{d} v$.
For $\sigma>0$ \eqref{main_eq} is called the Vlasov-Poisson-Fokker-Planck (VPFP) equation, for $\sigma=0$ it is called the Vlasov-Poisson (VP) equation.

Before we come to the central results of this paper let us first give a short overview of recent findings related to propagation of chaos in particular for singular interactions. Due to their wide applications in applied sciences, independent contributions from different research groups have  substantially contributed to the  developments of the field. Many new ideas and methods have been introduced to prove mean-field limit results in different settings and for different kinds of singular kernels.
A detailed review of the state of the art of  recent developments can for example be found in \cite{BDJ24,CHJ24,golse2016dynamics,jabin2017mean}. 

	Propagation of chaos for second order systems with singular interactions has been studied already by Oelschl\"ager in \cite{oelschlager91} for the case that $k^N$ is an approximation of the gradient of the Dirac-Delta distribution. This work is based on an estimate of a smoothed version of the modulated energy, which then implies convergence to the Euler system. 

The modulated energy method was developed for hydrodynamic limits, dating back to Dafermos \cite{Defermos79,Defermos20} in the context of conservation laws and Brenier  \cite{Brenier2000} in the  study of quasi-neutral limits from the Vlasov-Poisson equation to the incompressible Euler system for suitable initial data. Brenier's result \cite{Brenier2000}  was  later improved by Masmoudi for general initial data\cite{Masmoudi2001}. The modulated energy method has proven useful in various situations, for instance in the high field limit \cite{GNPS05}, the quasi-neutral limits \cite{Wang05,HLW06}, the quasi-neutral limit with vanishing viscosity  \cite{HLW08}, etc.  Also results in a hydrodynamic limit have to be mentioned. In a limiting regime where  the velocities are relaxed the modulated energy  method can be used to prove convergence of solutions of the kinetic Vlasov equation (or Vlasov-Fokker-Planck) to the Euler equation \cite{Kang18}. The same holds true for systems with velocity alignment \cite{KMT15}, in such situations even singular interactions like Coulomb can be treated \cite{CCJ21}.  
Even without noise, the modulated energy method can be used to treat singular interactions. Assuming mono-kineticity of the initial data as well as repulsiveness of the interaction Serfaty and Duerinckx \cite{serfaty2020mean} showed convergence of the Newtonian many-body system towards the pressureless Euler-Poisson system. 
Inspired by \cite{Brenier2000} and \cite{serfaty2020mean}, a combined mean-field and quasi-neutral limit, from Newton's law to incompressible Euler, has been proved in \cite{HKI21}. 

Parallel to this, the relative entropy method  was used by Jabin and Wang   to derive mean-field limit of the second order system for certain types of singular interactions  \cite{jabinWang2016,jabin2018quantitative}.  Later, the relative entropy method is combined with the modulated energy method to form the modulated free energy method to treat the mean field limit problem for stochastic systems with singular interactions in \cite{bresch2019mean}. 

In recent years, also methods which are not based on estimating the relative entropy between the microscopic and the effective description have been found to prove propagation of chaos.
Hauray and Jabin considered singular interaction forces scaling as  $1/\vert x \vert^{\lambda}$ in three dimensions with ${\lambda} < 1$ \cite{HJ07} as well as values for  ${\lambda}$ smaller  than but close to $2$, and a lower bound on the cut-off at $|x|= N^{-1/6}$  \cite{HJ15}. In their work they could quantify the rate of convergence in Wasserstein distance for sufficiently large enough $N$.   Kiessling in  \cite{K14} assumes no cut-off but some additional technical condition which can be read as a  bound on the maximal forces of the microscopic system along the trajectories and repulsiveness on the interaction.  

	By introducing an  $N$-dependent cut-off of the singularity, Lazarovic and Pickl proved a convergence in probability result  for the trajectories in \cite{LP}, which implies  weak propagation of chaos result for the VP case. Following these techniques,  Carrillo, Choi, and Salem could generalize this result for a system including white noise. In this case the limiting equation is Vlasov Poisson Focker Planck \cite{CCS19}. In addition they proved convergence of the marginals of the densities towards the limiting equation in Wasserstein distance rather explicitly by giving  a convergence rate estimate. In \cite{HLP20}, Huang, Liu, and Pickl improved the cut-off rate of the Coulomb potential for the VPFP case by using the regularization effect caused  by the noise term.

Recently, Bresch, Jabin, and Soler \cite{BJS23} developed a stability estimate for the error term in the BBGKY hierarchy for the VPFP case on a two dimensional torus. Bresch, Duerinckx, and Jabin \cite{BDJ24} applied the duality method and proved propagation of chaos (for second order system with singular potentials) in the sense of distributions.

	Note, that the different techniques we mentioned lead to different notions of convergence of the many-body system towards the effective equation. Convergence of the relative entropy of the respective densities directly implies convergence in a strong sense ($L^1$). However, all results we mentioned that control the relative entropy for a second order system with singular interaction were made in a special situation where the limiting description is given by Euler. 
Within the framework of modulated energy, the information on the mesoscopic level, i.e. the  kinetic Vlasov dynamics, is somehow missing, which is the main difference between first order and second order systems. However, the mesoscopic scaling or description is very important in many applications and shall not be bypassed.  
In contrast to that, the results we mentioned above that are not based on an estimate of the relative entropy obtained only the weak convergence (Wasserstein distance, convergence in distribution). 

The goal of this article is to show the strong form of propagation of chaos towards the VP (or VPFP) system. We will show that the quantitative estimates given by Huang, Liu, and Pickl \cite{HLP20} (for VPFP case) and Lazarovic and Pickl \cite{LP}  on the trajectory differences in the sense of probability are the right ingredient that enables us to follow Jabin and Wang \cite{jabinWang2016, jabin2018quantitative} and get an estimate on the relative entropy in the three dimensional case.
The estimates on the convergence in probability in \cite{HLP20} and \cite{LP} can be used in a rather direct manner to obtain the relative entropy estimates. We will therefore take over the setting and assumptions from \cite{HLP20} (VPFP case) and cite the relevant estimate. To avoid unnecessary repetitions we will not carry out the full proof in the VP case but only state the result. The proof is practically equivalent to the proof in the VPFP case.

For the force-term  in \eqref{micro_reg} we use the same version of a  regularized force $k^N$ as used in \cite{CCS19}, 
\begin{equation}
	\label{kN_LP}
	k^N(x)=\left\{\begin{array}{ll}
		k(x),  & \mbox{if\, } |x|\geq N^{-\delta}, \\[3mm]
		\pm C_d x N^{d\delta},  &  \mbox{if\, } |x|<N^{-\delta},
	\end{array}\right.
\end{equation}
where $\delta \in (0,1)$ will be given below.  We recall that the initial data of \eqref{micro_reg} is given as a set of i.i.d.  square integrable random variables as in \eqref{particleID} 
with $f_0$ be their common probability density function. 
We fix  a complete probability setting for \eqref{micro_reg}: let $(\Omega, \mathcal{F}, (\mathcal{F}_{t\geq 0}) , \mathbb{P})$ be a complete filtered probability space,
$d$-dimensional $\mathcal{F}_t$-Brownian motions $\{ (B^i_t)_{t\geq 0} \}_{i=1}^N$  be independent of each other, and let  the initial data $\{ (x^{i,N,\sigma_N}_0,v^{i,N,\sigma_N}_0) \}_{i=1}^N$ be independent of $\{ (B^i_t)_{t\geq 0} \}_{i=1}^N$.  

Since the interaction kernel $k^N$ now is regular enough, the Cauchy problem \eqref{micro_reg} together with the above given initial data has a unique square integrable solution  
$$
\Phi_t^N = (X_t, V_t) := ( x_t^{1,N,\sigma_N}, \dots,  x_t^{N,N,\sigma_N},  v_t^{1,N,\sigma_N}, \dots,  v_t^{N,N,\sigma_N}),
$$
and denote $f_t^{N}(x_1,\cdots,x_N,v_1,\cdots,v_N)$ by its joint distribution at time $t$. Note that our assumption \eqref{particleID} on the initial data implies $f_0^N = \bigotimes_{i=1}^N f_0(x_i, v_i)$.

We also present the $N$-particle relative entropy between the joint distribution $f_t^N$ and the tensorized law $f_t^{\otimes N}$, and also  the $k$-particle relative entropy between the $k$-marginal $f_t^{N, k}$ and $f_t^{\otimes k}$,  from \cite{jabinWang2016}, are defined as 
\begin{align*}
\mathcal{H}_N(f_t^{N} \ | \  f_t^{\otimes N})=\frac1N \mathcal{H}(f_t^{N} \ | \  f_t^{\otimes N})&:= \frac1N\int_{\mathbb{R}^{2dN}} f_t^{N} \log \frac{f_t^{N}}{ f_t^{\otimes N}} \ud z_1 \cdots \ud  z_N, \\
\mathcal{H}_k(f_t^{N,k} \ | \  f_t^{\otimes k})&:= \frac 1 k \int_{\R^{2dk}} f_t^{N,k} \log \frac{f_t^{N,k}}{ f_t^{\otimes k}}\ud z_1 \cdots \ud z_k, 
\end{align*}
respectively, where we used the short notation $z_i=(x_i, v_i) \in \R^d \times \R^d$.

\vskip5mm

We now come to the main results of our paper. Depending on different regularity assumptions on the initial datum we get convergence of the relative entropy of $f_t^N$ towards the product $f_t^{\otimes N}$ with different convergence rate (Theorem \ref{main_thm1} and Theorem \ref{main_thm2}). The $\sigma$-dependence of both these estimates is such that  one gets convergence for cases where, as $N\to\infty$, $\sigma$ does not tend to zero too fast.
Under further regularity assumption we are also able to handle the cases where  $\sigma$ tends to zero arbitrarily fast (Theorem \ref{cor_VP}).

\begin{theorem}\label{main_thm1}
Let $T>0$,  $\delta \in (0,1/d)$ be given and assume that the initial data $f_0\geq 0$ satisfies that 
\[
f_0 \in (L^1\cap L^\infty\cap \mathcal{P}_2)(\R^{2d}), \quad \int_{\R^{2d}} f_0 \log f_0\,\ud x \ud v <\infty, \quad (1+|v|^2)^{\frac{m_0}{2}} f_0 \in L^\infty(\R^{2d})\mbox{ for some } m_0>d,
\]
where $\mathcal{P}_2(\R^{2d})$ denotes the space of probability measures with finite 2nd moments. Let $f_t^N$ and $ f_t$ be the joint distribution to the particle system \eqref{micro_reg} with the regularized interaction \eqref{kN_LP} and the (unique) weak solution to the limit PDE \eqref{main_eq}, respectively. Assume further that $\sigma_N \equiv \sigma>0$.  Then there exists a positive constant $C=C(T)$, independent of $N$ and $\sigma$,  satisfying that 
\[
\sup_{0\le t \le T} \mathcal{H}_N (f_t^N \ | \  f_t^{\otimes N}) \le \frac{C\exp(C\sqrt{\log N})}{\sigma N^{2\delta}},
\]
for sufficiently large $N$.
\end{theorem}

\begin{remark}
Following the idea of \cite{jabinWang2016}, namely the combination of the classical Csisz\'ar-Kullback-Pinsker inequality (see \cite{Villani02}) and the sub-additivity of  the (scaled) relative entropy for instance in   \cite[Lemma 21]{CHH24}, we obtain immediately that  
\[
\|f_t^{N,k} - f_t^{\otimes k}\|_{L^1(\R^{2dk})}^2 \le  2 k \mathcal{H}_k (f_t^{N, k} \vert f_t^{\otimes k})  \leq 2 k  \mathcal{H}_N(f_t^{N} \ | \ f_t^{\otimes N}), 
\]
 where $f_t^{N,k}$ is the $k$-marginal  of $f_t^N$. We recall that since $f_t^N$ is  a symmetric probablity measure, its $k$-marginal $f_t^{N,k}$ is defined as 
 \[
 f_t^{N, k}(z_1, \cdots,  z_k) = \int_{(\mathbb{R}^{2d})^{N-k} }  f_t^N (z_1, \cdots, z_k, z_{k+1}, \cdots, z_N) \ud z_{k+1} \cdots \ud z_{N}.
 \] 
 In particular, we have the $L^1$-convergence of the first marginal $f_t^{N,1}$ to the (unique) weak solution $ f_t$ to \eqref{main_eq}, i.e.
 \begin{equation} \label{L1convergence}
 \sup_{t \in [0, T]} \|f_t^{N,1} - f_t\|_{L^1(\R^{2d})} \le \frac{C}{\sqrt{\sigma} N^{\delta-}}.
 \end{equation}
This implies of course further that
 \begin{equation} \label{L1convergencerho}
	\sup_{t \in [0, T]} \Big\|\displaystyle\int_{\R^d}f_t^{N,1}dv - \rho_t\Big\|_{L^1(\R^{d})} \le \frac{C}{\sqrt{\sigma} N^{\delta-}}.
\end{equation}

\end{remark}

In the case when $d=3$, following the smoothed setting of $k_N$ given by \cite{HLP20} and applying more regularity of the solution $f_t$ of \eqref{main_eq}, one can obtain a better convergence rate in $N$.  Let $\psi = \psi(x) \in \mathcal{C}^2(\R^3)$ which satisfies
$\mbox{supp } \psi(x) \subseteq B(0,1)$ with $\intr \psi(x)\ud x=1$. By using its rescaled function $\psi_\delta^N(x) := N^{3\delta} \psi(N^\delta x)$ like in \cite{HLP20}, we choose the different type of regularized version $k^N$ in \eqref{micro_reg} given by
\bq\label{kNHLP}
k^N(x) := k * \psi_\delta^N(x).
\eq
Then we have the following result with an improved convergence rate. 

\begin{theorem}
	\label{main_thm2}Let $d=3$ and $T>0$ be given. Assume that the initial data $f_0\geq 0$ satisfies the Assumption 1.1 in \cite{HLP20}, i.e.
\begin{align*}
   & f_0 \in (W^{1,1}\cap W^{1,\infty})(\R^6), \quad \intrr f_0 \log f_0\ud x \ud v <\infty, \\ 
   &(1+|v|^2)^{\frac{m_0}{2}} f_0 \in (W^{1,1}\cap W^{1,\infty})(\R^6)\mbox{ for some } m_0>6.\\
   & \mbox{There exists a constant } Q_v>0 \mbox{ such that } f_0(x,v)=0 \mbox{ for } |v|\geq Q_v.
\end{align*}
 Let $f_t^N$ and $ f_t$ be the joint distribution to the particle system \eqref{micro_reg} with the regularized interaction \eqref{kNHLP} and unique solution to \eqref{main_eq}, respectively.  Again assume that $\sigma_N \equiv \sigma>0$. Then for any  $\lambda_2 \in (3/10,1/3)$, there exists $\lambda_1 \in (0, \lambda_2/3)$ and positive constant $C=C(T)$,  independent of $N$ and $\sigma$, such that for  $\delta \in \lt[1/3, \min\lt\{ \frac{\lambda_1 + 3\lambda_2 +1}{6}, \frac{1-\lambda_2}{2} \rt\}\rt)$, it holds
 \[
	\sup_{0 \le t \le T} \mathcal{H}_N (f_t^N \ | \ f_t^{\otimes k})\leq \frac{C(\log N)^{3/2}}{\sigma N^{2\lambda_2}}, \quad \mbox{for sufficiently large } N.
\]
\end{theorem}

The third main result of this article  is to obtain the $L^1$ convergence of the marginal densities of the microscopic model \eqref{micro_reg} to the Vlasov-Poisson equation, i.e. \eqref{main_eq} with $\sigma=0$. 

\begin{theorem}\label{cor_VP}
In addition to those assumptions made in Theorem \ref{main_thm1}, we further assume that for $m>3$,
\[\begin{aligned} 
\nabla_{x,v}\log f_0 \in L^\infty(\R^6) \quad\mbox{ and }\quad\|\nabla \tilde \rho\|_{L^1(0,T;L^1\cap L^\infty)} <\infty.
\end{aligned}\]
Then there exists a positive constant $C=C(T)$,  independent of $N$ and $\sigma$,  that satisfies 
\[
\sup_{0\le t \le T} \mathcal{H}_N (f_t^N \ | \  \tilde f_t^{\otimes N}) \le C\sqrt\sigma\lt( \sqrt\sigma +1\rt)\exp\lt( C\sqrt{\log N}\rt) + CN^{-\delta}\log N, 
\]
where $\tilde f_t$ is the (unique) classical solution to \eqref{main_eq} in the case when  $\sigma=0$ with  the initial data $f_0$.
\end{theorem}

\begin{remark}
Actually the assumption on the boundedness of $\nabla\tilde\rho$ can be obtained in the following. Once we assume the initial data $f_0$ satisfies 
\[
(1+|v|^2)^{\frac{m}{2}} f_0,  \ (1+|v|^2)^{\frac{m}{2}}\nabla_{x,v} f_0 \in L^\infty(\R^6), \quad f_0 \in L^1 \cap W^{1,\infty}(\R^6),
\] 
we can use the method of characteristics to construct a classical solution $\tilde f_t$ to \eqref{VPeq} (locally in time for attactive case) satisfying 
\[
(1+|v|^2)^{\frac{m}{2}}\tilde f,  \ (1+|v|^2)^{\frac{m}{2}}\nabla_{x,v}\tilde f \in L^\infty(0,T;\R^6), \quad \tilde f \in L^\infty(0,T;L^1 \cap W^{1,\infty}(\R^6))
\]
and this naturally implies $\tilde \rho \in L^\infty(0,T;W^{1,\infty}(\R^6))$ since
\[
\|\nabla \tilde \rho_t\|_{L^\infty} \le \lt(\intr (1+|v|^2)^{-\frac{m}{2}}\,dv\rt)  \|(1+|v|^2)^{\frac{m}{2}}\nabla_x \tilde f_t\|_{L^\infty}.
\]

\end{remark}

The rest of the paper is organized as follows.  In Section 2,  we focus on the convergence towards Vlasov-Poisson-Fokker-Planck equation, namely we will give the proofs of theorems \ref{main_thm1} and \ref{main_thm2}, introduce the intermediate problem and cite the concrete result from \cite{HLP20}. We then proceed the relative entropy estimates in two steps: estimates between the solution of the Liouville equation and the intermediate solutions, and then the estimates between the intermediate solutions and  the real VPFP ones. Furthermore, we proceed the limit to the Vlasov-Poisson system in Section  3. Besides presenting the proof of Theorem \ref{cor_VP}, we give another proof for the $L^1$ convergence towards the Vlasov-Poisson system.

%%%%%%%%%%%%%%%%%%%%%%%%%%%%%%%%%%%%%%%%%%%%%%%%%%%%%%%%%%%%%%%%%%%%%%%%%%%%%%%%%5
%
%
%                        Section: Introduction 
%
%
%%%%%%%%%%%%%%%%%%%%%%%%%%%%%%%%%%%%%%%%%%%%%%%%%%%%%%%%%%%%%%%%%%%%%%%%%%%%%%%%%
\section{Propagation of chaos in $L^1$: Towards Vlasov-Poisson-Fokker-Planck}\label{sec:MF}
\setcounter{equation}{0}
In this section, we present the proof of Theorems \ref{main_thm1} and \ref{main_thm2}, which rely on  the relative entropy estimates for $N$-particle distribution functions of \eqref{micro_reg} and the $N$-copies of the Vlasov-Poisson-Fokker-Planck \eqref{main_eq}. We focus on  the Vlasov-Poisson-Fokker-Planck case and thus  in this section we all take $\sigma_N=\sigma >0$. 
By introducing an intermediate problem, the smoothed version of \eqref{main_eq} and its corresponding McKean-Vlasov process, and citing the concrete results from \cite{CCS19,LP} and \cite{HLP20} separately, we then obtain the relative entropy estimates in two steps: estimates between the solution of the Liouville equation and the intermediate solutions, and the estimates between the intermediate solutions and the solution of \eqref{main_eq}.

For simplicity, we drop everywhere the dependence on $\sigma$ and denote 
$$
\Phi_t^N = (X_t, V_t) := ( x_t^{1,N}, \dots,  x_t^{N,N},  v_t^{1,N}, \dots,  v_t^{N,N})
$$
to be the unique solution of \eqref{micro_reg}. 
The so-called intermediate particle system, the McKean-Vlasov process, has been introduced to do the key estimate. It is formally the McKean-Vlasov system or mean field system of \eqref{micro_reg}, i.e. for $i =1, 2, \cdots, N$, 
\begin{equation}\label{MV_reg}
\begin{cases} 
	\ud \bar x_t^{i,N} &=\bar v_t^{i, N} \ud t,\\
	\ud \bar v_t^{i,N} &= \int_{\R^d} k^N(\bar x_t^{i,N} - x)\bar \rho_t(x)\ud x \, \ud t +\sqrt{2\sigma} \ud B_t^i, 
\end{cases}
\end{equation}
where $ \bar\rho_t(x)=\int_{\R^d}\bar f_t(x,v) \ud v$ and $\bar f_t$ is the law of $(\bar x_t^{i,N},\bar v_t^{i,N})$.

It is well known that for fixed $N$ the above system with initial data \eqref{particleID} has a unique square integrable i.i.d. solution $(\bar x_t^{i,N},\bar v_t^{i,N})$, whose law satisfies the following Vlasov-Fokker-Planck equation with regularized interaction $k^N$,
\begin{equation}\label{VFP_reg}
\pa_t \bar f_t + v \cdot \nabla_x \bar f_t + k^N\star \bar\rho_t \cdot \nabla_v \bar f_t = \sigma\Delta_v\bar f_t
\end{equation}
that subjects to the initial data
\[
\bar f_0(x,v) = f_0(x,v).  
\]
We can thus construct an auxiliary flow 
\[
\Psi_t^N = (\bar X_t, \bar V_t) := (\bar x_t^{1,N}, \dots, \bar x_t^{N,N}, \bar v_t^{1,N}, \dots, \bar v_t^{N,N}), 
\]
and introduce
\[
\|\Phi_t^N - \Psi_t^N\|_\infty:= \sqrt{\log N}\|X_t - \bar X_t\|_\infty + \|V_t - \bar V_t\|_\infty, 
\]
as in \cite{LP}. 

\subsection{The proof of Theorem \ref{main_thm1}}
We consider the regularized kernel of type \eqref{kN_LP}. The following convergence in probability result has been obtained in  \cite[Lemma 3.2]{CCS19} (see also  \cite[Theorem 4.2]{LP}for the Vlasov-Poisson case). 

\begin{theorem}\label{conv_prob1}
Under the assumptions in Theorem \ref{main_thm1}, for any $\delta\in (0,\frac{1}{d})$ and $\alpha>0$, there exists $C(T,\alpha)$ such  that  for sufficiently large  $N$, 
\[
\mathbb{P}\bigg(\sup_{t \in [0,T]}\|\Phi_t^N - \Psi_t^N\|_\infty >N^{-\delta} \bigg)< C(T,\alpha)N^{-\alpha}. 
\]
\end{theorem}

In the next step, we will estimate the relative entropy  between the distribution $f_t^N$ of \eqref{micro_reg} and the tensorized law $\bar f_t^{\otimes N}$, where $\bar f_t$ solves  \eqref{VFP_reg}. We follow the framework provided in \cite{jabinWang2016} for second order systems. The $N-$particle  Liouville equation of the system \eqref{micro_reg} reads 
\begin{equation}\label{Liou}
\pa_t f_t^{N} + \sum_{i=1}^N v_i \cdot \nabla_{x_i} f_t^{N} + \sum_{i=1}^N \bigg(\frac{1}{N-1}\sum_{j\neq i}k^N(x_i-x_j)\cdot \nabla_{v_i} f_t^{N} \bigg) = \sigma \sum_{i=1}^N\Delta_{v_i} f_t^{N}.
\end{equation}
Since for any fixed $N$, the interaction force $k^N$ is at least $C_b^1$,  the above Liouville equation has a unique entropy solution in the sense for instance in \cite[Proposition 2.1]{jabinWang2016}. Furthermore, the evolution equation satisfied by  the tensorized law $\bar f_t^{\otimes N}=\bar f_t(x_1, v_1)\cdots \bar f_t(x_N, v_N)$ reads 
\begin{equation}\label{NVlasov}
\pa_t \bar f_t^{\otimes N} + \sum_{i=1}^N v_i \cdot \nabla_{x_i} \bar f_t^{\otimes N} + \sum_{i=1}^N \lt(k^N\star \bar\rho_t(x_i)\cdot \nabla_{v_i} \bar f_t^{\otimes N} \rt) = \sigma\sum_{i=1}^N\Delta_{v_i} \bar f_t^{\otimes N}.
\end{equation}
Before we present our detailed proof, let us review several known facts. 

Following the  strategy as in the proof of   \cite[Proposition 7.2]{LP}\label{law_large_num},  we recall the  special form of  {\em  Law of Large Numbers} as  \cite[Lemma 2.2]{CCS19}. 
\begin{proposition}
Let $(Y_1, \dots, Y_N)$ be i.i.d. random variables with common  law $\rho\in (L^1\cap L^\infty)(\R^d)$ and define the associated empirical measure $\rho_N := \frac1N \sum_{i=1}^N \delta_{Y_i}$. Suppose that $h:\R^d \to \R^d$ satisfies that 
\[
|h(x)| \le c_0 \min\{N^{\kappa\delta}, |x|^{-\kappa}\}\quad\mbox{for some } \kappa,\delta>0, \, \mbox{and \, } \kappa<d, 
\]
and choose that $\e = 2\kappa\delta + (1-d\delta )\mathds{1}_{1>d\delta} \in (0,2)$. Then for any $m >\frac{1}{2-\e}$, there exists $C_m>0$ such that 
\[
\mathbb{E} \bigg[ \sup_{1 \le i \le N} |h\star\rho_N (Y_i) - h\star\rho(Y_i)|^{2m} \bigg] \le C_m N^{1-(2-\e)m}.
\]

\end{proposition}

To compare the McKean-Vlasov system for the regularized Vlasov-Fokker-Planck equation with that for \eqref{main_eq}, we introduce

\begin{equation}\label{MV_lim}
\begin{cases} 
	\ud  x_t &= v_t  \ud t,\\
	\ud  v_t &= \intr k( x_t - x) \rho_t(x)\, \ud x \ud t +\sqrt{2\sigma} \ud B_t^i, 
\end{cases}
\end{equation}
where $ \rho_t(x)=\int_{\R^3} f_t(x,v) \ud v$ and $f_t$ is the (unique) solution to \eqref{main_eq}. 

\begin{proposition}\cite[Proposition 3.2]{CCS19}\label{prop_diff}
Let $(\bar x_t, \bar v_t)$ and $(x_t, v_t)$ be solutions to \eqref{MV_reg} (with the regularized kernel \eqref{kNHLP} ) and \eqref{MV_lim}, respectively, given the same independent  and identically distributed initial data. Assume further that $\|\bar\rho_t\|_{L^1(0,T; L^1\cap L^\infty)}, \|\rho_t\|_{L^1(0,T;L^1\cap L^\infty)} <\infty$.  Then it holds for large enough $N$ that
\[
\mathbb{E}\bigg[\sup_{0\le t\le T}(\sqrt{\log N} |\bar x_t - x_t| + |\bar v_t - v_t|)^2 \bigg]^{1/2}\le CN^{-\delta }\exp (C\sqrt{\log N}).
\]
\end{proposition}

\noindent
Finally, we recall the following technical lemma that is essentially  \cite[Lemma 6.3]{LP} or \cite[Lemma 3.1]{CCS19}. 
\begin{lemma}\label{kernel_diff}
For any $\xi \in \R^d$ with $|\xi| <2N^{-\delta}$, it holds that
\[
|k^N(x) - k^N(x+\xi)| \le \ell^N(x) |\xi|,
\]
where $\ell^N(x)$ is given by
\[
\ell^N(x) := \left\{\begin{array}{ll}
		\frac{C(d)}{|x|^d}, & \mbox{if}\,  |x|\ge N^{-\delta}, \\
		C(d)N^{d\delta}, &\mbox{if}\,    |x|< N^{-\delta}. 
	\end{array}\right.
\]
\end{lemma}

\noindent
Now we proceed to the proof of Theorem \ref{main_thm1}.

\begin{proof}[Proof of Theorem \ref{main_thm1}]
We compute the time evolution of the relative entropy as follows. For simplicity, we omit the integral domain which is $\R^{2dN}$. We adopt the notations that $z_i = (x_i, v_i) \in \R^d \times \R^d$ and $Z_N= (z_1, \cdots z_N) \in \R^{2dN}$.  We proceed as 
\begin{align*}
&\frac{d}{dt}\mathcal{H}_N(f_t^{N} \ | \ f_t^{\otimes N}) \\
= &\frac1N\int  \pa_t f_t^{N} \log \frac{f_t^{N}}{f_t^{\otimes N}}\,dZ_N  - \frac1N \int f_t^{N}  \pa_t \log  f_t^{\otimes N}\,dZ_N\\
%=& - \frac\sigma N \iint_{\R^{3N}\times \R^{3N}} %\sum_{i=1}^N  \frac{ %|\nabla_{v_i}f_t^{N}|^2}{f_t^{N}}\,dz_N + %\frac1N\iint_{\R^{3N}\times \R^{3N}} \lt(\sum_{i=1}^N v_i %\cdot \nabla_{x_i} f_t^{N} \rt)\log \bar f_t^{\otimes %N}\,dz_N\\
%&\quad -\frac1N\iint_{\R^3\times \R^3}\frac{1}{N-1} %\sum_{j\neq i} k^N(x_i-x_j)  \cdot \nabla_{v_i} f_t^{N}  %\log \bar f_t^{\otimes N}\,dz_N\\
%&\quad + \frac{\sigma}{N}\iint_{\R^{3N}\times \R^{3N}}  %\sum_{i=1}^N \frac{\nabla_{v_i}  f_t^{N} \cdot %\nabla_{v_i} \bar f_t^{\otimes N}}{\bar f_t^{\otimes %N}}\,dz_N + \frac1N\iint_{\R^{3N}\times \R^{3N}} %\lt(\sum_{i=1}^N v_i \cdot \nabla_{x_i} \bar f_t^{\otimes %N}\rt)\frac{f_t^{N}}{\bar f_t^{\otimes N}}\,dz_N\\
%&\quad + \frac1N\iint_{\R^{3N}\times %\R^{3N}}\lt[\sum_{i=1}^N  k^N \star \bar \rho_t(x_i) %\cdot \nabla_{v_i} \bar f_t^{\otimes N} \rt] %\frac{f_t^{N}}{\bar f_t^{\otimes N}}\, dz_N\\
%&\quad -\frac\sigma N \iint_{\R^{3N}\times \R^{3N}} %\sum_{i=1}^N \frac{f_t^{N} \Delta_{v_i} \bar f_t^{\otimes %N}}{\bar f_t^{\otimes N}}\,dz_N\\
=& -\frac{\sigma}{N}\sum_{i=1}^N \int  f_t^{N} \lt|\nabla_{v_i} \log \lt(\frac{f_t^{N}}{f_t^{\otimes N}} \rt) \rt|^2 \,dZ_N\\
&\quad -\frac1 N \sum_{i=1}^N  \int f_t^{N} \,  \bigg[  \bigg(\frac{1}{N-1}\sum_{j\neq i} k^N(x_i-x_j) - k \star \rho_t(x_i)\bigg) \cdot  \nabla_{v_i}\log \lt(\frac{f_t^{N}}{ f_t^{\otimes N}} \rt) \bigg]  \,dZ_N. 
\end{align*}
Applying Cauchy-Schwarz's inequality to the last term, one obtains that 
\[
\begin{split}
 \frac{d}{dt}\mathcal{H}_N(f_t^{N} \vert \ f_t^{\otimes N}) &\le -\frac{\sigma }{2N}  \sum_{i=1}^N \int f_t^{N}   \lt|\nabla_{v_i} \log \lt(\frac{f_t^{N}}{ f_t^{\otimes N}} \rt) \rt|^2 \,dZ_N\\
&\quad + \frac{1}{2N\sigma} \sum_{i=1}^N  \int f_t^{N} \bigg| \frac{1}{N-1}\sum_{j\neq i} k^N(x_i-x_j) - k \star  \rho_t(x_i)\bigg|^2  \,dZ_N.
\end{split}
\]
Using the exchangeability of the particle system \eqref{micro_reg}, one can rewrite the inequality above as 
\[\begin{aligned}
&\frac{d}{dt}\mathcal{H}_N (f_t^{N} \ | \ f_t^{\otimes N}) + \frac{\sigma }{2N} \sum_{i=1}^N \int f_t^{N}\,  \lt|\nabla_{v_i} \log \lt(\frac{f_t^{N}}{ f_t^{\otimes N}} \rt) \rt|^2 \,dZ_N\\
&\le \frac{1}{2\sigma } \mathbb{E}\bigg[ \bigg| k \star  \rho_t(x_t^{1,N}) - \frac{1}{N-1} \sum_{j=2}^{N}  k^N (x_t^{1,N} - x_t^{j,N}) \bigg|^2 \bigg]=: \frac{1}{2\sigma} \mathbb{M},
\end{aligned}\]
where the expectation is taken over the law $f_t^N$. 
We then split the estimate for $\mathbb{M}$ into several parts as follows, 

\[\begin{aligned}
\mathbb{M} &\le   5 \mathbb{E}\Big[ \lt| k \star \rho_t(x_t^{1,N}) -k \star \rho_t ( \bar x_t^{1,N}) \rt|^2 \Big]\\
&\quad +5\mathbb{E}\Big[ \lt| k \star \rho_t(\bar x_t^{1,N}) -k \star \bar\rho_t ( \bar x_t^{1,N}) \rt|^2 \Big]\\
&\quad + 5\mathbb{E}\Big[ \lt| (k-k^N) \star \bar \rho_t(\bar x_t^{1,N}) \rt|^2 \Big]\\
&\quad +  5 \mathbb{E}\bigg[\,  \bigg| k^N \star \bar \rho_t(\bar x_t^{1,N}) - \frac{1}{N-1} \sum_{j\neq 1} k^N (\bar x_t^{1,N} - \bar x_t^{j,N}) \bigg|^2 \, \bigg]\\
&\quad +  5 \mathbb{E}\bigg[ \bigg|\frac{1}{N-1} \sum_{j\neq 1} k^N (\bar x_t^{1,N} - \bar x_t^{j,N}) -\frac{1}{N-1} \sum_{j\neq 1} k^N ( x_t^{1,N} - x_t^{j,N})  \bigg|^2 \bigg]=: \sum_{i=1}^5 \mathbb{M}_i.
\end{aligned}\]
To estimate $\mathbb{M}_1$, we first recall the following estimate for instance in \cite{CJpre}. Given a Newtonian or Coulombian kernel $k(x) = \pm C_d \frac{x}{|x|^d}$, there exists a universal constant $C$, such that for any $x_1, x_2 \in \R^d$, 
\bq\label{log_lip}
\int_{\R^d} |k (x_1 - y) - k (x_2 - y)| |h(y)|\,dy \le C\|h\|_{(L^1\cap L^\infty)(\R^d)}|x_1-x_2|(1-\log^-|x_1-x_2|),
\eq
 where  $\log^- r := \min\{0, \log r\}$.
 
We define  the event $\mathcal{A}_T$ as
\[
\mathcal{A}_T := \{\sup_{0\le t \le T} \|\Phi_t^N - \Psi_t^N\|_\infty \le N^{-\delta} \}. 
\]
Then   Theorem \ref{conv_prob1} implies that  for any $\alpha >0$, 
\[\begin{aligned}
\mathbb{M}_1 &\le  5 \mathbb{E}\lt[ \lt| k \star \rho_t(x_t^{1,N}) -k \star  \rho_t (\bar x_t^{1,N})  \rt|^2 \bigg| \mathcal{A}_T^c \rt]\\
&\quad + 5  \mathbb{E}\lt[  \lt| k \star  \rho_t(x_t^{1,N}) -k \star \rho_t (\bar x_t^{1,N})  \rt|^2 \bigg| \mathcal{A}_T \rt]\\
&\le 20 \|k \star \rho_t\|_{L^\infty(\R^d)}^2 \mathbb{P}(\mathcal{A}_T^c)+ C\|{\rho_t}\|^2_{(L^1\cap L^\infty)(\R^d)}\frac{N^{-2\delta}}{\log N}\Big(1-\log\frac{N^{-\delta}}{\sqrt{\log N}}\Big)^2\\
&\le C  N^{-\alpha}+ C N^{-2\delta}\log N\le CN^{-2\delta}\log N,
\end{aligned}\]
where we have used the fact that $\alpha>0$ is arbitrary and that $r(1-\log r)\leq r_0(1-\log r_0)$ for $0<r=|x_1-x_2|<r_0=\frac{N^{-\delta}}{\sqrt{\log N}}<1$.\\

\noindent For $\mathbb{M}_2$, one observes that
\[
\mathbb{M}_2\le C\|\bar \rho_t\|_{L^\infty} \|k\star (\bar\rho_t - \rho_t)\|_{L^2}^2.
\]
Since $k = \pm \nabla W$, where $W$ is the fundamental solution to the Poisson equation, i.e. $-\Delta W = \delta_0$, one has
\[
\|k\star (\bar\rho_t - \rho_t)\|_{L^2}^2 \le \|\bar\rho_t - \rho_t\|_{\dot{H}^{-1}}^2 \le C (\|\bar\rho_t\|_{ L^\infty} + \|\rho_t\|_{ L^\infty}) W_2^2(\bar\rho_t, \rho_t),
\]
where we used \cite[Theorem 2.9]{Loe06}. Since Proposition \ref{prop_diff} implies
\[
 W_2(\bar\rho_t, \rho_t) \le CN^{-\delta} \exp(C\sqrt{\log N})
\]
we get
\[
\mathbb{M}_2 \le C\|\bar \rho_t\|_{L^\infty}(\|\bar\rho_t\|_{L^\infty} + \|\rho_t\|_{L^\infty}) N^{-2\delta} \exp(C\sqrt{\log N}). 
\]

\noindent For $\mathbb{M}_3$, we note that $k^N(x) = k(x)$ if $|x| \ge N^{-\delta}$. Thus
\[\begin{aligned}
|(k^N-k)\star \bar\rho_t (x)| &= \lt|\int_{\R^d} (k^N(x-y)-k(x-y))\bar\rho_t(y)\,dy \rt|\\
&\le C\lt|\int_{|x-y|\le N^{-\delta}} \lt(  N^{d\delta}|x-y| + \frac{1}{|x-y|^{d-1}} \rt) \bar\rho_t(y)\,dy \rt|\\
&\le C\|\bar\rho_t\|_{L^\infty}\int_{|z|\le N^{-\delta}} \lt(N^{(d-1)\delta} + \frac{1}{|z|^{d-1}}\rt)\,dz\le C\|\bar\rho_t\|_{L^\infty}N^{-\delta}.
\end{aligned}\]
Hence,
\[
\mathbb{M}_3\le C\|\bar\rho_t\|^2_{L^\infty}N^{-2\delta}.
\]

\noindent For $\mathbb{M}_4$, we apply 
$h(x) = k^N(x)$ and $ \kappa=d-1$
to Proposition \ref{law_large_num} and choose $m>\frac{1}{1-\delta}$ sufficiently large so that
\[\begin{aligned}
	\mathbb{M}_4 &\le C  \mathbb{E}\bigg[ \bigg| k^N \star \bar \rho_t(\bar x_t^{1,N}) - \frac{1}{N-1} \sum_{j\neq 1} k^N (\bar x_t^{1,N} - \bar x_t^{j,N}) \bigg|^{2m}  \bigg]^{\frac 1m}\\
	&\le C_m  N^{-1+(d-2)\delta+\frac 1m}\le C N^{-2\delta}.
\end{aligned}
\]	

\noindent
Finally, for $\mathbb{M}_5$, we use Lemma \ref{kernel_diff} and Theorem \ref{conv_prob1} , with the choice $\alpha=d$ and $\ell^N(0)=0$, to yield 

\[\begin{aligned}
\mathbb{M}_5 & \le 5  \mathbb{E}\bigg[ \bigg| \frac{1}{N-1} \sum_{j\neq 1} k^N (\bar x_t^{1,N} - \bar x_t^{j,N}) -\frac{1}{N-1} \sum_{j\neq 1} k^N ( x_t^{1,N} - x_t^{j,N})  \bigg|^2 \bigg| \mathcal{A}_T^c \bigg] \\
&\quad + 5 \mathbb{E}\bigg[\bigg| \frac{1}{N-1} \sum_{j\neq 1} k^N (\bar x_t^{1,N} - \bar x_t^{j,N}) -\frac{1}{N-1} \sum_{j\neq 1} k^N ( x_t^{1,N} - x_t^{j,N})  \bigg|^2 \bigg| \mathcal{A}_T \bigg]\\
&\le C   \|k^N\|_{L^\infty(\R^d)}^2\mathbb{P}(\mathcal{A}_T^c) + C\mathbb{E}\lt[ \lt|\frac{N^{-\delta}}{\sqrt{\log N}(N-1)}\sum_{j\neq 1} \ell^N(\bar x_t^{1,N}- \bar x_t^{j,N}) \rt|^2\bigg| \mathcal{A}_T\rt]\\
&\le C(\alpha) N^{2(d-1)\delta-\alpha} + CN^{-2\delta}\mathbb{E}\bigg[ \bigg|\frac{1}{N-1}\sum_{j\neq 1} \ell^N(\bar x_t^{1,N}- \bar x_t^{j,N}) \bigg|^2\bigg]\\
&\le C N^{-2\delta} + CN^{-2\delta}\frac{1}{(N-1)^2}\sum_{i,j\neq 1}\mathbb{E}\bigg[ \ell^N(\bar x_t^{1,N}- \bar x_t^{i,N})\ell^N(\bar x_t^{1,N}- \bar x_t^{j,N}) \bigg]\\
&\le C N^{-2\delta} + CN^{-2\delta}\iiint \ell^N(x-y)\ell^N(x-z)\bar\rho_t(y)\bar\rho_t(z)\bar\rho_t(x) dydzdx\\
&\le C N^{-2\delta} + C N^{-2\delta}\|\ell^N \star\bar\rho_t\|_{L^\infty}^2\le C N^{-2\delta}(\log N)^2,
\end{aligned}\]
where we have used
\[\begin{aligned}
|\ell^N\star\bar\rho_t| &\le C\|\bar\rho_t\|_{L^\infty}\lt(\int_{\{|x|\le N^{-\delta}\}} N^{d\delta}\,dx + \int_{\{N^{-\delta}<|x|\le 1\}} \frac{1}{|x|^d}\,dx\rt) + C\|\bar\rho_t\|_{L^1}\\
&\le C(1+\log N).
\end{aligned}\]

\noindent
Therefore, we get the estimate for relative entropy
\[\begin{aligned}
\frac{d}{dt}&\mathcal{H}_N(f_t^{N} \ |  f_t^{\otimes N}) + \frac{\sigma}{2N}\iint \sum_{i=1}^N \lt|\nabla_{v_i} \log \lt(\frac{f_t^{N}}{ f_t^{\otimes N}} \rt) \rt|^2 f_t^{N}\,dZ_N\\
&\leq \frac{CN^{-2\delta}}{\sigma}(\exp\lt(C\sqrt{\log N}\rt) + (\log N)^2)\leq \frac{C}{\sigma}N^{-2\delta}\exp\lt(C\sqrt{\log N}\rt),
\end{aligned}\]
Together with the fact that the initial relative entropy $\mathcal{H}(f_0^{\otimes N}\ | \ f_0^{\otimes N})=0$ and $e^{C\sqrt{\log N}} \gg  (\log N)^2$ if $N$ is sufficiently large, we deduce our desired result.

\end{proof}

\subsection{Proof of Theorem \ref{main_thm2}} 
 Here, we provide the enhanced convergence rate of system \eqref{micro_reg} with the different choice of the regularized kernel \eqref{kNHLP} towards \eqref{main_eq}. Overall settings and arguments are almost the same as the previous section, but
%
%In \cite{HLP20}, an intermediate interacting particle system has been introduced to do the key estimate. It is formallly the McKean-Vlasov system or mean field system of \eqref{micro_reg}, 
%\bq\label{MV_reg}
%\begin{aligned} 
%	d\bar x_t^{i,N} &=\bar v_t^{i, N} dt,\\
%	d\bar v_t^{i,N} &= \intr k^N(\bar x_t^{i,N} - x)\bar \rho_t(x)\,dxdt +\sqrt{2\sigma} dB_t^i.
%\end{aligned}
%\eq
%where $ \bar\rho_t(x)=\int_{\R^3}\bar f_t(x,v)dv$ and $\bar f_t$ is the law of $(\bar x_t^{i,N},\bar v_t^{i,N})$.
%It is well known that for fixed $N$ the above system with initial data \eqref{particleID} has a unique square integrable i.i.d. solution $(\bar x_t^{i,N},\bar v_t^{i,N})$, whose law satisfies the following Vlasov-Fokker-Planck equation with regularized interaction,
%\bq\label{VFP_reg}
%\pa_t \bar f_t + v \cdot \nabla_x \bar f_t + k^N\star \bar\rho_t \cdot \nabla_v \bar f_t = \sigma_N\Delta_v\bar f_t
%\eq
%that subjects to initial data
%\[
%\bar f_0(x,v) = f_0(x,v).  
%\]
%We can thus construct an auxiliary flow 
%$$
%\Psi_t^N = (\bar X_t, \bar V_t) := (\bar x_t^{1,N}, \dots, \bar x_t^{N,N}, \bar v_t^{1,N}, \dots, \bar v_t^{N,N}).
%$$ 
%By introducing
%\[
%\|\Phi_t^N - \Psi_t^N\|_\infty:= \sqrt{\log N}\|X_t - \bar X_t\|_\infty + \|V_t - \bar V_t\|_\infty,
%\]
the difference is that we use the following convergence in probability result  given in \cite{HLP20}. 

\begin{theorem}\cite[Theorem 1.2]{HLP20}\label{conv_prob2}
For any $T>0$, assume that flows $\Phi_t^N = (X_t, V_t)$ and $\Psi_t^N = (\bar X_t, \bar V_t)$ satisfy the regularized interacting particle system \eqref{micro_reg} and  the correspoding $N-$independent copies of the McKean-Vlasov process \eqref{MV_reg} respectively,  with the regularized kernel \eqref{kNHLP} and the same initial data $\Phi_0^N = \Psi_0^N$ sampled from $(f_0)^{\otimes N}$. Moreover, let the assumptions on $f_0$ in Theorem \ref{main_thm2} hold. Then for any $\alpha>0$ and $\lambda_2\in (0,1/3)$, there exist some $\lambda_1 \in (0, \lambda_2/3)$ and $N_0(\alpha)$,  such that for $N\ge N_0$,  the following estimate holds with the cut-off index $\delta \in [1/3, \min\{ (\lambda_1 + 3\lambda_2+1)/6, (1-\lambda_2)/2\})$, 
\[
\mathbb{P}\lt(\max_{t \in [0,T]}\|\Phi_t^N - \Psi_t^N\|_\infty >N^{-\lambda_2}\rt) < N^{-\alpha}.
\]
\end{theorem}

\noindent
In addition, we need an analog of Proposition \ref{prop_diff} due to the different choices of our regularized kernel \eqref{kNHLP}, namely the following lemma. 

\begin{lemma}\label{lem_diff3}
Let $(\bar x_t, \bar v_t)$ and $(x_t, v_t)$ be solutions to \eqref{MV_lim} and \eqref{MV_reg} with the regularized kernel \eqref{kNHLP}, respectively, corresponding to the same independent and  identically distributed initial data. Once we have $\|\bar\rho_t\|_{L^1(0,T; L^1\cap L^\infty)}, \|\rho_t\|_{L^1(0,T;L^1\cap L^\infty)} <\infty$, we get
\[
\mathbb{E}\lt[\sup_{0\le t\le T}(\sqrt{\log N} |\bar x_t - x_t| + |\bar v_t - v_t|)^2 \rt]^{1/2}\le CN^{-\delta }\exp (C\sqrt{\log N}).
\]
\end{lemma}
\begin{proof}
Although the proof can be deduced  from almost the same argument as in the proof of  Proposition \ref{prop_diff} from \cite[Proposition 3.2]{CCS19}, we provide it here for self-containedness. We split the proof into three steps as follows:\\

\noindent $\bullet$ (Step 1) First, we observe that  the inequality
\[
|k^N(x) - k^N(y)| \le C|x-y| \lt(\frac{1}{\max\{|x|, N^{-\delta}\}^3} + \frac{1}{\max\{|y|,  N^{-\delta}\}^3} \rt),
\]
holds, where $C$ is independent of $x$, $y$ and $N$. When $|x|, |y|>N^{-\delta}$, then one uses the  fact that $k^N(z) = k(z)$ when  $|z| \ge N^{-\delta}$ (see \cite[Lemma 2.1, (i)]{HLP20}) to yield

\bq\label{ineq_k}
|k^N(x) - k^N(y)| = |k(x) - k(y)| \le C|x-y| \lt(\frac{1}{|x|^3} + \frac{1}{|y|^3} \rt).
\eq
When $|x|, |y|\le N^{-\delta}$, then we use $\|\nabla k^N\|_{L^\infty} \le CN^{3\delta}$ (also see \cite[Lemma 2.1, (ii)]{HLP20}) to have
\[
|k^N(x) - k^N(y)| \le \|\nabla k^N\|_{L^\infty}|x-y| \le CN^{3\delta} |x-y|.
\]
When $|x| \le N^{-\delta} <|y|$ or $|y|\le N^{-\delta} <|x|$, we only consider the case $|x| \le N^{-\delta} <|y|$ since the other case is similar. Then we choose a point $\tilde x$, which is an intersection between the line segment connecting $x$ and $y$, and the closed ball $\bar B(0,N^{-\delta})$. Then we can use the estimates in previous cases to attain
\[\begin{aligned}
|k^N(x) - k^N(y)| &\le |k^N(x) - k^N(\tilde x)| + |k^N (\tilde x) - k^N(y)|\\
&\le \|\nabla k^N\|_{L^\infty}|x-\tilde x| + |k(\tilde x) - k(y)|\\
&\le CN^{3\delta }|x-\tilde x| + C|\tilde x - y| \lt(\frac{1}{|\tilde x|^3} + \frac{1}{|y|^3} \rt)\\
&\le C|x-y| \lt(\frac{1}{N^{-3\delta}} + \frac{1}{|y|^3} \rt),
\end{aligned}\] 
where we also used $|x-\tilde x| , |\tilde x -y| \le |x-y|$. Thus, it implies the desired result.\\

\noindent $\bullet$ (Step 2) We set $D_t^N := \sqrt{\log N} |\bar x_t -x_t| + |\bar v_t - v_t|$ and consider an independent copy$(y, \bar y)$ of $(x, \bar x)$. Then we get
\bq\label{est_diff_proc}
\mathbb{E} \lt[ |k^N(x_t-y_t) - k^N(\bar x_t - \bar y_t)| D_t^N\rt] \le C\sqrt{\log N } \lt(\|\bar\rho_t\|_{L^1\cap L^\infty}  + \|\rho_t\|_{L^1\cap L^\infty} \rt)\mathbb{E}[(D_t^N)^2].
\eq
For this, we use the estimates in Step 1 to get

\[\begin{aligned}
\mathbb{E}& \lt[ |k^N(x_t-y_t) - k^N(\bar x_t - \bar y_t)| D_t^N\rt] \\
&\le C \mathbb{E} \lt[\Big(|x_t -\bar x_t| + |y_t - \bar y_t |\Big) \lt(\frac{1}{\max\{|x_t-y_t|,N^{-\delta}\}^3 } + \frac{1}{\max\{|\bar x_t - \bar y_t|, N^{-\delta}\}^3} \rt) D_t^N \rt]\\
&\le C\mathbb{E} \lt[|x_t -\bar x_t| \lt(\frac{1}{\max\{|x_t-y_t|,N^{-\delta}\}^3 } + \frac{1}{\max\{|\bar x_t - \bar y_t|, N^{-\delta}\}^3} \rt) D_t^N \rt]\\
&\quad + C\mathbb{E} \lt[|y_t -\bar y_t|  \lt(\frac{1}{\max\{|x_t-y_t|,N^{-\delta}\}^3 } + \frac{1}{\max\{|\bar x_t - \bar y_t|, N^{-\delta}\}^3} \rt) D_t^N \rt]=: I_1+I_2.
\end{aligned}\]
For $I_1$, one has

\[\begin{aligned}
I_1 &\le C\mathbb{E}_{(x_t, v_t, \bar x_t, \bar v_t)}\lt[|x_t - \bar x_t| D_t^N \mathbb{E}_{(y_t, \bar y_t)} \lt[ \frac{1}{\max\{|x_t-y_t|,N^{-\delta}\}^3 } + \frac{1}{\max\{|\bar x_t - \bar y_t|, N^{-\delta}\}^3} \rt] \rt],
\end{aligned}\]
and we can get
\[\begin{aligned}
&\mathbb{E}_{(y_t, \bar y_t)} \lt[ \frac{1}{\max\{|x_t-y_t|,N^{-\delta}\}^3 } + \frac{1}{\max\{|\bar x_t - \bar y_t|, N^{-\delta}\}^3} \rt]\\
&\le \intr \frac{1}{\max\{|x_t-y|,N^{-\delta}\}^3 }\rho_t(y)\,dy + \intr \frac{1}{\max\{|\bar x_t - y|, N^{-\delta}\}^3}\bar\rho_t (y)\,dy\\
&\le \int_{|x_t-y|\le N^{-\delta}}N^{3\delta} \rho_t(y)\,dy  + \int_{N^{-\delta}<|x_t-y|\le 1} \frac{1}{|x_t -y|^3}\rho_t(y)\,dy + \|\rho_t\|_{L^1}\\
&\quad + \int_{|\bar x_t-y|\le N^{\-\delta}}N^{3\delta} \bar\rho_t(y)\,dy  + \int_{N^{-\delta}<|\bar x_t-y|\le 1} \frac{1}{|\bar x_t -y|^3}\bar \rho_t(y)\,dy + \|\bar\rho_t\|_{L^1}\\
&\le C(1+ \log N) (\|\bar\rho_t\|_{L^1 \cap L^\infty} + \|\rho_t\|_{L^1\cap L^\infty}).
\end{aligned}\]
Hence, we have
\[
I_1 \le C(\log N) \mathbb{E}\lt[ |x_t - \bar x_t| D_t^N\rt] \le C\sqrt{\log N}  \mathbb{E}[(D_t^N)^2].
\]
For $I_2$, we use Cauchy-Schwarz inequality, Young's inequality and the above estimates to have
\[\begin{aligned}
I_2 &\le C\sqrt{\log N}\mathbb{E}\lt[ |y_t -\bar y_t|^2  \lt(\frac{1}{\max\{|x_t-y_t|,N^{-\delta}\}^3 } + \frac{1}{\max\{|\bar x_t - \bar y_t|, N^{-\delta}\}^3} \rt) \rt]\\
&\quad + \frac{C}{\sqrt{\log N}}\mathbb{E}\lt[(D_t^N)^2   \lt(\frac{1}{\max\{|x_t-y_t|,N^{-\delta}\}^3 } + \frac{1}{\max\{|\bar x_t - \bar y_t|, N^{-\delta}\}^3} \rt) \rt]\\
&\le C\sqrt{\log N } \mathbb{E}[(D_t^N)^2],
\end{aligned}\]
and combine the estimates for $I_1$ with those for $I_2$ to yield \eqref{est_diff_proc}.\\

\noindent $\bullet$ (Step 3) Finally, we prove the desired result. We again adopt the notation $D_t^N$ and define
\[
\mathscr{D}_t^N := \sup_{0\le s \le t}(D_s^N)^2.
\]
Then one uses It\^o's lemma to get

\[\begin{aligned}
\mathscr{D}_t^N  &\le 2 \int_0^t  \lt(\sqrt{\log N} |\bar v_s - v_s| + |k^N \star \bar \rho (\bar x_s) - k \star \rho (x_s)| \rt) D_s^N \,ds\\
&\le 2\int_0^t \sqrt{\log N}(D_s^N)^2\,ds + 2\int_0^t  |k^N \star \bar \rho (\bar x_s) - k^N \star \rho (x_s)| D_s^N \,ds\\
&\quad + 2\int_0^t |(k^N - k)\star \rho(x_s)| D_s^N\,ds.
\end{aligned}\]
By taking expectation, one obtains

\[\begin{aligned}
\mathbb{E}[\mathscr{D}_t^N] &\le 2\sqrt{\log N}\int_0^t \mathbb{E} [\mathscr{D}_s^N]\,ds  + 2\int_0^t \mathbb{E} \lt[ |k^N \star \bar \rho (\bar x_t) - k^N \star \rho (x_t)| D_s^N \rt]\,ds\\
&\quad + 2\int_0^t \mathbb{E} \lt[|(k^N - k)\star \rho(x_s)| D_s^N\rt]\,ds\\
&=: 2\sqrt{\log N}\int_0^t \mathbb{E} [\mathscr{D}_s^N]\,ds + J_1 +J_2.
\end{aligned}\]
For $J_1$, we use the estimates in step 2 to have
\[
\begin{aligned}
J_1 &\le C\sqrt{\log N}\int_0^t (\|\bar \rho_s\|_{L^1\cap L^\infty} + \|\rho_s\|_{L^1\cap L^\infty}) \mathbb{E}[(D_s^N)^2]\,ds.
\end{aligned}
\]
For $J_2$, we note that $k^N(x) = k(x)$ if $|x| \ge N^{-\delta}$. Thus
\[\begin{aligned}
|(k^N-k)\star \rho_s (x)| &= \lt|\int_{\R^3} (k^N(x-y)-k(x-y))\rho_s(y)\,dy \rt|\\
&= \lt|\int_{|x-y|\le N^{-\delta}} \lt( \int_{|z|\le N^{-\delta}}(k(x-y-z)-k(x-y))\psi_{\delta}^N(z)\,dz\rt) \rho_s(y)\,dy \rt|\\
&\le \int_{|z|\le N^{-\delta}}\lt(\int_{|x-y|\le N^{-\delta}} |k(x-y-z)-k(x-y)|\rho_s(y)\,dy\rt)\psi_\delta^N(z)\,dz\\
&\le C\|\rho_s\|_{L^\infty}\int_{|z|\le N^{-\delta}} \lt( \int_{|y|\le 2N^{-\delta}}|k(y)|\,dy + \int_{|y|\le N^{-\delta}} |k(y)|\,dy \rt)\psi_\delta^N(z)\,dz\\
&\le C\|\rho_s\|_{L^\infty}N^{-\delta}.
\end{aligned}\]
Hence,
\[
J_2\le C\int_0^t \|\rho_s\|_{L^\infty} \mathbb{E}[(D_s^N)^2]\,ds + CN^{-2\delta}\int_0^t \|\rho_s\|_{L^\infty}\,ds.
\]
Thus, we gather all the results and apply Gr\"onwall's lemma to yield
\[
\mathbb{E}[\mathscr{D}_t^N] \le CN^{-2\delta}\exp(C\sqrt{\log N} t),
\]
where $C= C(T, \|\bar\rho\|_{L^1(0,T;L^1\cap L^\infty)}, \|\rho\|_{L^1(0,T;L^1\cap L^\infty)}, \delta) $ is independent of $N$.
\end{proof} 

\vskip3mm
Now, we are ready to prove Theorem \ref{main_thm2}.

\begin{proof}[Proof of Theorem \ref{main_thm2}]
Again, we compute the time evolution of the relative entropy as in Theorem \ref{main_thm1} and for simplicity, we also omit the integral domain. Then the same computation gives

\[\begin{aligned}
&\frac{d}{dt}\mathcal{H}_N (f_t^{N} \ | \ f_t^{\otimes N}) + \frac{\sigma }{2N} \sum_{i=1}^N \int f_t^{N}\,  \lt|\nabla_{v_i} \log \lt(\frac{f_t^{N}}{ f_t^{\otimes N}} \rt) \rt|^2 \,dZ_N\\
&\le \frac{1}{2\sigma } \mathbb{E}\bigg[ \bigg|\Big( k \star  \rho_t(x_t^{1,N}) - \frac{1}{N-1} \sum_{j=2}^{N}  k^N (x_t^{1,N} - x_t^{j,N})\Big) \bigg|^2 \bigg]=: \frac{1}{2\sigma } \mathbb{L},
\end{aligned}\]
where the expectation is taken over the law $f_t^N$.

We also split the estimate for $\mathbb{L}$ into several parts as follows: 

\[\begin{aligned}
\mathbb{L} &\le   5 \mathbb{E}\Big[ \lt| k \star \rho_t(x_t^{1,N}) -k \star \rho_t ( \bar x_t^{1,N}) \rt|^2 \Big]\\
&\quad +5\mathbb{E}\Big[ \lt| k \star \rho_t(\bar x_t^{1,N}) -k \star \bar\rho_t ( \bar x_t^{1,N}) \rt|^2 \Big]\\
&\quad + 5\mathbb{E}\Big[ \lt| (k-k^N) \star \bar \rho_t(\bar x_t^{1,N}) \rt|^2 \Big]\\
&\quad +  5 \mathbb{E}\bigg[\,  \bigg| k^N \star \bar \rho_t(\bar x_t^{1,N}) - \frac{1}{N-1} \sum_{j\neq 1} k^N (\bar x_t^{1,N} - \bar x_t^{j,N}) \bigg|^2 \, \bigg]\\
&\quad +  5 \mathbb{E}\bigg[ \bigg|\frac{1}{N-1} \sum_{j\neq 1} k^N (\bar x_t^{1,N} - \bar x_t^{j,N}) -\frac{1}{N-1} \sum_{j\neq 1} k^N ( x_t^{1,N} - x_t^{j,N})  \bigg|^2 \bigg]=: \sum_{i=1}^5 \mathbb{L}_i.
\end{aligned}\]

For $\mathbb{L}_1$, We define  the event $\mathcal{A}_T$ as
\[
\mathcal{A}_T := \{\max_{0\le t \le T} \|\Phi_t^N - \Psi_t^N\|_\infty \le N^{-\lambda_2} \}. 
\]
Then by the same argument for $\mathbb{M}_1$ in Theorem \ref{main_thm1},  for any $\alpha >0$, 
\[\begin{aligned}
\mathbb{L}_1 &\le  5 \mathbb{E}\lt[ \lt| k \star \rho_t(x_t^{1,N}) -k \star  \rho_t (\bar x_t^{1,N})  \rt|^2 \bigg| \mathcal{A}_T^c \rt]\\
&\quad + 5  \mathbb{E}\lt[  \lt| k \star  \rho_t(x_t^{1,N}) -k \star \rho_t (\bar x_t^{1,N})  \rt|^2 \bigg| \mathcal{A}_T \rt]\le CN^{-2\lambda_2}\log N.
\end{aligned}\]

\noindent For $\mathbb{L}_2$, we again follow the same argument for $\mathbb{M}_2$ in the proof of Theorem \ref{main_thm1} and use Lemma \ref{lem_diff3} to yield
\[
\mathbb{L}_2 \le C\|\bar \rho_t\|_{L^\infty}(\|\bar\rho_t\|_{L^\infty} + \|\rho_t\|_{L^\infty})^2 N^{-2\delta} \exp(C\sqrt{\log N})\le CN^{-2\lambda_2}. 
\]

\noindent For $\mathbb{L}_3$, we can use the arguments employed for $J_2$ in Step 3 of Lemma \ref{lem_diff3} to get
\[
\mathbb{L}_3\le C\|\bar\rho_t\|^{2}_{L^\infty}N^{-2\delta}.
\]

\noindent For $\mathbb{L}_4$, one uses the law of large numbers in \cite[Lemma 2.6]{HLP20} that we rephrase it as follows. For any  $\alpha>0$, there exists $C_{1,\alpha}>0$,  such that  for any $t\in [0,T]$, the following event  $\mathcal{B}^\alpha_t$, 
\[\begin{aligned}
	\mathcal{B}^\alpha_t &:= \bigg\{ \max_{1\le i \le N}\bigg| k^N \star \bar \rho_t(\bar x_t^{i,N}) - \frac{1}{N-1} \sum_{j\neq i} k^N (\bar x_t^{i,N} - \bar x_t^{j,N})\bigg|\ge C_{1,\alpha}N^{2\delta-1}\log N \bigg\},
\end{aligned}\]
has probability 
\[\begin{aligned}
	\mathbb{P}(\mathcal{B}^\alpha_t) \le N^{-\alpha}.
\end{aligned}\]
Therefore, $\mathbb{L}_4$ can be estimated in the following
\[\begin{aligned}
	\mathbb{L}_4 &\le  \mathbb{E}\bigg[ \bigg| k^N \star \bar \rho_t(\bar x_t^{1,N}) - \frac{1}{N-1} \sum_{j\neq 1} k^N (\bar x_t^{1,N} - \bar x_t^{j,N}) \bigg|^2 \bigg| (\mathcal{B}_t^\alpha)^c \bigg]\\
	&\quad + \mathbb{E}\bigg[ \bigg|  k^N \star \bar \rho_t(\bar x_t^{1,N}) - \frac{1}{N-1} \sum_{j\neq 1} k^N (\bar x_t^{1,N} - \bar x_t^{j,N}) \bigg|^2 \bigg| \mathcal{B}_t^\alpha\bigg]\\
	&\le   C^2_{1,\alpha}N^{2(2\delta-1)}(\log N)^2+ \|k^N\|^2_{L^\infty(\R^3)}N^{-\alpha}\\
	&\le C N ^{2(2\delta-1)}(\log N)^2\le CN^{-2\lambda_2}
\end{aligned}
\]	
for sufficiently large $N$, where we used $2\delta-1 <-\lambda_2$ and we have chosen $\alpha>2$, since $\|k_N\|_{L^\infty(\R^3)}\leq C N^{2\delta}$ (\cite[Lemma 2.1]{HLP20}), so that the second term is smaller than the first term. \\

\noindent
Finally, for $\mathbb{L}_5$, we introduce the event $\mathcal{S}_T(\Lambda)$ as 
\[\begin{aligned}
\mathcal{S}_T(\Lambda) &:= \bigg\{ \max_{1\le i \le N}\bigg|\frac{1}{N-1} \sum_{j\neq i} k^N (\bar x_t^{i,N} - \bar x_t^{j,N}) -\frac{1}{N-1} \sum_{j\neq i} k^N ( x_t^{i,N} - x_t^{j,N}) \bigg|\\
&\hspace{5mm} \le \Lambda (\log N) \|X_t - \bar X_t\|_\infty + \Lambda (\log N)^2 \lt(N^{6\delta - 1-\lambda_1-4\lambda_2} + N^{3\lambda_1-2\lambda_2} + N^{2\delta-1} \rt), \, \forall t \in [0,T] \bigg\}.
\end{aligned}\]
Then the result \cite[Proposition 3.2]{HLP20} implies, there exists $C_{\lambda_2} >0$,  such that
\[
\mathbb{P}(\mathcal{A}_T \cap \mathcal{S}_T^c(C_{\lambda_2})) \le N^{-4 \delta - 2 \lambda_2}.
\]
%{\bf Zhenfu: Is there a mistake here? If let us assume this is true, then to get  the estimate for $M_3$, we will in the end get $\lambda_2 >1$.   }

Thus, since $\|k_N\|_{L^\infty(\R^3)}\leq C N^{2\delta}$ (\cite[Lemma 2.1]{HLP20}), we have

\[\begin{aligned}
\mathbb{L}_5 & \le 5  \mathbb{E}\bigg[ \bigg| \frac{1}{N-1} \sum_{j\neq 1} k^N (\bar x_t^{1,N} - \bar x_t^{j,N}) -\frac{1}{N-1} \sum_{j\neq 1} k^N ( x_t^{1,N} - x_t^{j,N})  \bigg|^2 \bigg| \mathcal{A}_T^c \bigg] \\
&\quad + 5 \mathbb{E}\bigg[\bigg| \frac{1}{N-1} \sum_{j\neq 1} k^N (\bar x_t^{1,N} - \bar x_t^{j,N}) -\frac{1}{N-1} \sum_{j\neq 1} k^N ( x_t^{1,N} - x_t^{j,N})  \bigg|^2 \bigg| \mathcal{A}_T \cap \mathcal{S}_T^c(C_{\lambda_2})\bigg]\\
&\quad +  5 \mathbb{E}\bigg[ \bigg| \frac{1}{N-1} \sum_{j\neq 1} k^N (\bar x_t^{1,N} - \bar x_t^{j,N}) -\frac{1}{N-1} \sum_{j\neq 1} k^N ( x_t^{1,N} - x_t^{j,N})  \bigg|^2  \bigg| \mathcal{A}_T \cap \mathcal{S}_T(C_{\lambda_2})\bigg]\\
&\le C   \|k^N\|_{L^\infty(\R^3)}^2\Big( \mathbb{P}(\mathcal{A}_T^c)+\mathbb{P}(\mathcal{A}_T \cap \mathcal{S}_T^c(C_{\lambda_2}))\Big)\\
&\quad + 2(C_{\lambda_2})^2 (\log N)^{\frac{3}{2}} N^{-2\lambda_2}  +  2(C_{\lambda_2})^2  (\log N)^4 \lt(N^{6\delta - 1-\lambda_1-4\lambda_2} + N^{3\lambda_1-2\lambda_2} + N^{2\delta-1} \rt)^2\\
&\le CN^{-2\lambda_2} (\log N)^{3/2}.
\end{aligned}\]
where we have used the fact $\lambda_1 \in (0, \lambda_2/3)$ and $\delta \in [1/3, \min\{ (\lambda_1 + 3\lambda_2+1)/6, (1-\lambda_2)/2\})$.\\

%{\bf Zhenfu: We should double check here. }

Therefore, we get the estimate for relative entropy
\[\begin{aligned}
\frac{d}{dt}&\mathcal{H}(f_t^{N} \ | \ \bar{f}_t^{\otimes N}) + \frac{\sigma}{2N}\iint_{\R^{3N}\times \R^{3N}} \sum_{i=1}^N \lt|\nabla_{v_i} \log \lt(\frac{f_t^{N}}{\bar f_t^{\otimes N}} \rt) \rt|^2 f_t^{N}\,dz_N\\
&\leq \frac{C}{\sigma}(N^{-2\lambda_2}(\log N)^{3/2}+N^{2(2\delta-1)}(\log N)^2)\leq \frac{C}{\sigma}N^{-2\lambda_2}(\log N)^{3/2},
\end{aligned}\]
where we chose $\delta<\frac{1-\lambda_2}{2}$. Together with the fact that the initial relative entropy $\mathcal{H}(f_0^{\otimes N}\ | \ f_0^{\otimes N})=0$, we finish the proof.

\end{proof}

%%%%%%%%%%%%%%%%%%%%%%%%%%%%%%%%%%%%%%%%%%%%%%%%%%%%%%%%%%%%%%%%%%%
%
%
%   Section 3: Conv. to VP
%
%
%
%%%%%%%%%%%%%%%%%%%%%%%%%%%%%%%%%%%%%%%%%%%%%%%%%%%%%%%%%%%%%%%%%%%

\section{Propagation of chaos in $L^1$: Towards Vlasov-Poisson}\label{MF:VP}
\setcounter{equation}{0}
In this section, we study the convergence from $N$-particle system \eqref{micro_reg} \eqref{kN_LP} to the Vlasov-Poisson equation  in $3$ dimension. After a short preparation given in subsection 3.1, the subsection 3.2 is dedicated to the proof of Theorem \ref{cor_VP}. Since our aim is to get the propagation of chaos result in the $L^1$ sense, we will give an alternative proof of it without using the result in Theorem \ref{cor_VP}. Actually, we can obtain it by a combination of Theorem \ref{main_thm1} and the relative entropy estimate between solutions of Vlasov-Poisson-Fokker-Planck and Vlasov-Poisson. The details of this alternative proof is given in subsection 3.3.

\subsection{Technical lemmas}
In this subsection, we provide technical lemmas for the proof of Theorem \ref{cor_VP}. For this, we consider the trajectories 

\bq\label{char_VP}
\begin{aligned}
d\tilde{x}_t &= \tilde{v}_t\,dt,\\
d\tilde{v}_t & = \intr k(\tilde{x}_t - x)\tilde\rho_t(x)\,dx dt,
\end{aligned}
\eq
to the Vlasov-Poisson equation
\bq\label{VPeq}
\pa_t \tilde{f}_{t} + v\cdot \nabla \tilde{f}_{t} + k \star \tilde\rho_{t} \cdot \nabla_v \tilde{f}_{t} = 0
\eq
with the McKean-Vlasov system  \eqref{MV_lim}. We first provide the following $L^\infty$-estimate.

\begin{lemma}\label{logf_bd}
For $T>0$, let $\tilde f_t$ be a smooth solution to \eqref{VPeq} on $[0,T]$ decaying fast at infinity satisfying 
\[
\|\nabla \tilde \rho\|_{L^1(0,T;L^1\cap L^\infty)} <\infty.
\]
\noindent
Then we have
\[
\|\nabla_{x,v} \log \tilde f_t\|_{L^\infty}\le C\|\nabla_{x,v} \log \tilde f_0\|_{L^\infty},
\]
where $C= C(T, \|\nabla \tilde \rho\|_{L^1(0,T;L^1\cap L^\infty)})$ is a positive constant.
\end{lemma}

\begin{remark}
Instead of assuming boundedness of the gradient \(\nabla_{x,v} \log \tilde f_0\), one can relax this requirement by imposing polynomial growth conditions. For example, consider bounds of the form  
\[
|\nabla_{x,v} \log \tilde f_0(x, v)| \leq C(1 + |x|^\alpha + |v|^\alpha),
\]  
where \(C > 0\) and \(\alpha \geq 0\) are constants. Such conditions naturally arise in cases involving Gaussian-type distributions in $v$, where the logarithmic gradient exhibits polynomial growth in \(v\). With this relaxed condition, one should modify the following proofs with additional technical efforts. 

\end{remark}

\begin{proof}
From \eqref{VPeq}, one has

\[
\pa_t \log \tilde f_t + v \cdot \nabla \log \tilde f_t + k\star\tilde\rho_{t} \cdot \nabla_v \log \tilde f_t = 0.
\]
For $i,j=1,2,3$, we can get

\[
\pa_t (\pa_i\log \tilde f_t) + v \cdot \nabla (\pa_i \log \tilde f_t) + k\star \tilde\rho_t \cdot \nabla_v (\pa_i \log \tilde f_t) = -\pa_i (k\star \tilde\rho_t)\cdot \nabla_v \log \tilde f_t
\]
and
\[
\pa_t (\pa_{v_j} \log \tilde f_t) +v \cdot \nabla (\pa_{v_j} \log \tilde f_t)+ k\star {\tilde\rho_t \cdot \nabla_v } (\pa_{v_j} \log \tilde f_t)=- \pa_j \log \tilde f_t .
\]
Then along the characteristics, one obtains

\[
\frac{d}{dt}\|\nabla_v \log \tilde f_t\|_{L^\infty} \le \|\nabla \log \tilde f_t\|_{L^\infty}
\]
and
\[
\frac{d}{dt}\|\nabla_x \log \tilde f_t\|_{L^\infty}\le \|\nabla k\star\tilde \rho_t\|_{L^\infty} \|\nabla_v \log \tilde f_t\|_{L^\infty}
\]
and since we have (see \cite{B77} for example)

\bq\label{pot_contr}
\|\nabla k\star \tilde \rho_t \|_{L^\infty} \le C\lt[ (1+ \|\tilde \rho_t\|_{L^\infty})(1+\log (1+\|\nabla \tilde \rho_t\|_{L^\infty}) + \|\tilde \rho_t\|_{L^1}\rt],
\eq
we can use our assumption to have
\[
\frac{d}{dt}\|\nabla_{x,v}\log \tilde f_t\|_{L^\infty} \le C \|\nabla_{x,v}\log \tilde f_t\|_{L^\infty},
\]
which implies our desired estimate.

\end{proof}

\begin{remark}
Suppose $\tilde f_0$ is of the form
\[
\tilde f_0(x,v) = \frac{1}{(1+|x|^2)^\alpha} \cdot \frac{1}{(1+|v|^2)^\beta}
\] 
with $\alpha, \beta>0$ sufficiently large. Then
\[
|\nabla_x \log \tilde f_0| = \lt|\frac{2\alpha x}{1+|x|^2}\rt|, \quad |\nabla_v \log \tilde f_0| = \lt|\frac{2\beta v}{1+|v|^2}\rt|
\]
and hence we can have initial data satisfying our desired conditions.
\end{remark}
Now, we compare the trajectory \eqref{char_VP} with the Mckean-Vlasov equation \eqref{MV_lim} for the Vlasov-Poisson-Fokker-Planck equation.

\begin{lemma}\label{lem_diff2}
Let $(\tilde x_t, \tilde v_t)$ and $(x_t, v_t)$ be solutions to \eqref{char_VP} and \eqref{MV_lim} with $\sigma=\sigma_N$, respectively, corresponding to the same independent and identically distributed initial data. Once we have $\|\tilde\rho\|_{L^{1}(0,T; L^1\cap L^\infty)}, \|{\rho}\|_{L^1(0,T;L^1\cap L^\infty)} <\infty$, we get
\[
\sup_{0\le t\le T}\mathbb{E}\lt[\sqrt{\log N} |\tilde x_t - x_t| + |\tilde v_t - v_t| \rt]\le C\lt(\sqrt{\sigma_N} + N^{-\theta}(\log N) \rt)e^{C\sqrt{\log N}},
\]
where  $\theta>0$ is any constant and $C = C(\theta, T, \|\tilde\rho\|_{L^1(0,T;L^1\cap L^\infty)}, \|\rho\|_{L^1(0,T;L^1\cap L^\infty)})$ is a positive constant independent of $N$.
\end{lemma}

\begin{proof}
We first have

\[\begin{aligned}
d(\tilde x_t - x_t) &= (\tilde v_t - v_t)dt\\
d(\tilde v_t - v_t) &=\lt[ k\star \tilde \rho_{t} (\tilde x_t) -  k\star \rho_{t}(x_t) \rt]dt - \sqrt{2\sigma_N} dB_t.
\end{aligned}\]
and define $\mathfrak{D}(t) := \mathbb{E}\lt[\sqrt{\log N} |\tilde x_t - x_t| + |\tilde v_t - v_t| \rt]$. Then Young's inequality implies

\[\begin{aligned}
\mathbb{E}\lt[ |\tilde x_t - x_t|\rt] &\le \int_0^t  \mathbb{E}\lt[|\tilde v_s - v_s| \rt]\,ds,\\
\mathbb{E}\lt[|\tilde v_t -v_t| \rt] &\le \int_0^t \mathbb{E} \lt[ |k\star\tilde\rho_{s}(\tilde x_s) - k\star\rho_{s}(x_s)| \rt]\,ds + \sqrt{2\sigma_N} \mathbb{E}\lt[ \lt| B_t \rt|\rt]\\
&\le \int_0^t \mathbb{E} \lt[ |k\star\tilde\rho_{s}(\tilde x_s) - k\star\rho_{s}(\tilde{x}_s)| \rt]\,ds\\
&\quad + \int_0^t \mathbb{E} \lt[ |k\star\rho_{s}(\tilde x_s) - k\star\rho_{s}(x_s)| \rt]\,ds + \sqrt{2\sigma_N t}\\
&\le \int_0^t \mathbb{E} \lt[ |k\star\tilde\rho_{s}(\tilde x_s) - k\star\rho_{s}(\tilde{x}_s)| \rt]\,ds\\
&\quad + \int_0^t \mathbb{E} \lt[ |k\star\rho_{s}(\tilde x_s) - k\star\rho_{s}(x_s)| \rt]\,ds + \sqrt{2\sigma_N t}\\
&=: L_1 + L_2 +  \sqrt{2\sigma_{N} t}.
\end{aligned}\]
To estimate $L_1$, we set an independent copy $(y_s, \tilde{y}_s)$ of $(x_s, \tilde{x}_s)$. Then we obtain
\[\begin{aligned}
L_1 &= \int_0^t \mathbb{E} \lt[ \mathbb{E}_{y_s, \tilde{y}_s}[|k(\tilde{y}_s - \tilde{x}_s) - k(y_s - \tilde{x}_s)| ]  \rt]\,ds\\
&= \int_0^t \mathbb{E} \lt[ \mathbb{E}_{y_s, \tilde{y}_s}\lt[|k(\tilde{y}_s - \tilde{x}_s) - k(y_s - \tilde{x}_s)| \ \Big | \ |\tilde{y}_s - y_s|>N^{-\theta} \rt]  \rt]\,ds\\
&\quad + \int_0^t \mathbb{E} \lt[\mathbb{E}_{y_s, \tilde{y}_s}\lt[|k(\tilde{y}_s - \tilde{x}_s) - k(y_s - \tilde{x}_s)| \ \Big | \ |\tilde{y}_s - y_s|\le N^{-\theta}\rt]  \rt]\,ds\\
&=: L_{11} + L_{12}.
\end{aligned}\]
For $L_{11}$, note that $|\tilde{y}_s - \tilde{x}_s| <N^{-\theta}/2$ implies $|y_s -\tilde{x}_s|\ge N^{-\theta}/2$ and $|y_s -\tilde{x}_s|< N^{-\theta}/2$ implies $|\tilde{y}_s - \tilde{x}_s|\ge N^{-\theta}/2$. So one gets, for any $\tilde{x}_s \in \R^3$,
\[
\{ (y_s, \tilde{y}_s) \ | \ |\tilde{y}_s - y_s|>N^{-\theta}, |\tilde{y}_s - \tilde{x}_s| <N^{-\theta}/2 \} \cap \{ (y_s, \tilde{y}_s) \ | \ |\tilde{y}_s - y_s|>N^{-\theta}, |y_s - \tilde{x}_s| <N^{-\theta}/2 \} = \emptyset
\]
Hence we use \eqref{ineq_k} to obtain

\begin{align*}
L_{11}&=\int_0^t\Bigg( \iiint_{\{|\tilde{y} -y|>N^{-\theta}, \ |\tilde{y} - \tilde{x}| < N^{-\theta}/2\}} + \iiint_{\{|\tilde{y}-y|>N^{-\theta}, \ |y - \tilde{x}| < N^{-\theta}/2\}} \\
&\hspace{0.6cm}- \iiint_{\{|y -\tilde{y}|>N^{-\theta}, \ \min\{|\tilde{y} - \tilde{x}|, |y-\tilde{x}|\} \ge N^{-\theta}/2\}}\Bigg) |k(\tilde{y} - \tilde{x}) - k(y - \tilde{x})| \tilde\rho_s(\tilde{y})\rho_s(y)\tilde\rho_s(\tilde x)\,d\tilde{y} dy d\tilde x ds\\
&\le \int_0^t \iiint_{\{|\tilde{y} -y|>N^{-\theta}, \ |\tilde{y} - \tilde{x}| < N^{-\theta}/2\}} \frac{|\tilde{y} - y|}{|\tilde{y} - \tilde{x}|^3}\tilde\rho_s(\tilde y) \rho_s(y)\tilde\rho_s(\tilde x)\,d\tilde y dy d\tilde xds\\
&\quad +  \int_0^t \iiint_{\{|\tilde{y} -y|>N^{-\theta}, \ |y - \tilde{x}| < N^{-\theta}/2\}} \frac{|\tilde{y} - y|}{|y - \tilde{x}|^3}\tilde\rho_s(\tilde y) \rho_s(y)\tilde\rho_s(\tilde x)\,d\tilde y dy d\tilde xds\\
&\quad + \int_0^t \iiint_{\{|y -\tilde{y}|>N^{-\theta}, \ \min\{|\tilde{y} - \tilde{x}|, |y-\tilde{x}|\} \ge N^{-\theta}/2\}}|k(\tilde{y} - \tilde{x}) - k(y - \tilde{x}) |\tilde\rho_s(\tilde{y})\rho_s(y)\tilde\rho_s(\tilde x)\,d\tilde{y} dy d\tilde xds\\
&\le C\log N \int_0^t \|\tilde\rho_s\|_{L^\infty}\iint_{\{ |\tilde y - y| >N^{-\theta}\}}|\tilde y - y| \tilde\rho_s(\tilde y)\rho_s(y)\,d\tilde y dy ds\\
&\quad + \int_0^t \iiint_{\{|y -\tilde{y}|>N^{-\theta}, \ N^{-\theta}/2\le\min\{|\tilde{y} - \tilde{x}|, |y-\tilde{x}|\} \le 1\}} + \int_0^t \iiint_{\{|y -\tilde{y}|>N^{-\theta}, \min\{|\tilde{y} - \tilde{x}|, |y-\tilde{x}|\} > 1\}}\\
&\qquad|\tilde y - y| \lt(\frac{1}{|\tilde y - \tilde x|^3} + \frac{1}{|y - \tilde x|^3} \rt)\tilde\rho_s(\tilde{y})\rho_s(y)\tilde\rho_s(\tilde x)\,d\tilde{y} dyd\tilde x ds\\
&\le C\log N \int_0^t  \|\tilde\rho_s\|_{L^1 \cap L^\infty}\iint_{\{ |\tilde y - y| >N^{-\theta}\}}|\tilde y - y| \tilde\rho_s(\tilde y)\rho_s(y)\,d\tilde y dy ds\\
&\le C\sqrt{\log N} \int_0^t \|\tilde\rho_s\|_{L^1\cap L^\infty} \mathfrak{D}(s)\,ds.
\end{align*}
For $L_{12}$, we use \eqref{log_lip} to yield

\[\begin{aligned}
L_{12}&= \int_0^t \iiint_{\{|\tilde y - y| \le N^{-\theta}\}}(k(\tilde y - \tilde x) - k(y - \tilde x))\tilde\rho_s (\tilde x) \tilde\rho_s(\tilde y) \rho_s(y)\,d\tilde y dy d\tilde x ds\\
&\le \int_0^t \iint_{\{|\tilde y - y|\le N^{-\theta}\}} \lt|\intr |k(\tilde y - \tilde x ) - k(y - \tilde x)| \tilde\rho_s (\tilde x)\,d\tilde x \rt|\tilde\rho_s(\tilde y) \rho_s(y)d\tilde y dy ds\\
&\le C\int_0^t \|\tilde\rho_s\|_{L^1\cap L^\infty}\iint_{\{|\tilde y - y|\le N^{-\theta}\}} |\tilde y - y| (1- \log^- |\tilde y - y|)\tilde\rho_s(\tilde y) \rho_s(y)\,d\tilde y dy ds\\
&\le C N^{-\theta}\log N \int_0^t \|\tilde\rho_s\|_{L^1\cap L^\infty}\,ds,
\end{aligned}\]
where $N$ and $\theta$ are chosen so that $N^{\theta} \ge e$.\\

\noindent For $L_2$, we similarly proceed as before to yield

\begin{align*}
L_2 &\le \int_0^t \mathbb{E} \lt[ |k\star\rho_{s}(\tilde x_s) - k\star\rho_{s}(x_s)|  \ \Big | \ |\tilde x_s - x_s| > N^{-\theta}\rt]\,ds\\
&\quad +\int_0^t \mathbb{E} \lt[ |k\star\rho_{s}(\tilde x_s) - k\star\rho_{s}(x_s)|  \ \Big | \ |\tilde x_s - x_s| \le N^{-\theta}\rt]\,ds\\
&\le C\sqrt{\log N}\int_0^t \|\rho_{s}\|_{L^1\cap L^\infty} \mathfrak{D}(s)\,ds + CN^{-\theta} (\log N)\int_0^t \|\rho_{s}\|_{L^1\cap L^\infty} \,ds.
\end{align*}
Thus, we gather all the estimates for $L_1$ and $L_2$ to yield

\[\begin{aligned}
\mathfrak{D}(t) &\le C\sqrt{\log N}\int_0^t \lt(\|\tilde\rho_{s}\|_{L^1\cap L^\infty}+ \|\rho_{s}\|_{L^1\cap L^\infty} \rt)\mathfrak{D}(s)\,ds\\
&\quad + CN^{-\theta} (\log N)\int_0^t \lt(\|\tilde\rho_{s}\|_{L^1\cap L^\infty}+ \|\rho_{s}\|_{L^1\cap L^\infty} \rt) \,ds + \sqrt{2\sigma_N t},
\end{aligned}\]
and use Gronwall's lemma to obtain

\[\begin{aligned}
\mathfrak{D}(t) &\le \sqrt{2\sigma_N t} + CN^{-\theta}(\log N) \int_0^t \lt(\|\tilde\rho_{s}\|_{L^1\cap L^\infty}+ \|\rho_{s}\|_{L^1\cap L^\infty} \rt) \,ds\\
&\quad + \int_0^t \lt(\sqrt{2\sigma_N s} + CN^{-\theta}(\log N)\int_0^s \lt(\|\tilde\rho_{\tau}\|_{L^1\cap L^\infty}+ \|\rho_{\tau}\|_{L^1\cap L^\infty} \rt)\,d\tau  \rt)\\
&\hspace{1.4cm}\times \exp\lt(C\sqrt{\log N}\int_s^t\lt(\|\tilde\rho_{\tau}\|_{L^1\cap L^\infty}+ \|\rho_{\tau}\|_{L^1\cap L^\infty} \rt)\,d\tau \rt)\,ds\\
&\le C\lt( \sqrt{\sigma_N} + N^{-\theta}(\log N)\rt) e^{C\sqrt{\log N}},
\end{aligned}\]
where $C = C(T, \|\tilde\rho\|_{L^1(0,T;L^1\cap L^\infty)}, \|\rho\|_{L^1(0,T;L^1\cap L^\infty)})$ is a positive constant independent of $N$ and this is our desired result.

\end{proof}

\subsection{Proof of Theorem \ref{cor_VP}}
Now we proceed to the proof of Theorem \ref{cor_VP}. Similarly to the previous proofs, we estimate
\[\begin{aligned}
\frac{d}{dt}\mathcal{H}_N(f_t^N | \tilde f_t^{\otimes N}) &= -\frac{\sigma_N}{N}\sum_{i=1}^N \int_{\R^{6N}} f_t^N |\nabla_{v_i} \log f_t^N|^2  \,dZ_N\\
&\quad - \frac1N\sum_{i=1}^N\int_{\R^{6N}}f_t^N \lt[\lt(\frac{1}{N-1}\sum_{j\neq i} k^N (x_i - x_j)  - k\star\tilde\rho_{t}(x_i)\rt)\cdot \nabla_{v_i}\log \tilde f_t^{\otimes N}\rt] \,dZ_N \\
&\quad + \frac{\sigma_N}{N} \int_{\R^{6N}} f_t^N \nabla_{v_i} \log f_t^N \nabla_{v_i} \log \tilde f_t^{\otimes N}\,dZ_N\\
&\le  -\frac{\sigma_N}{2N} \sum_{i=1}^N\int_{\R^{6N}} |\nabla_{v_i} \log f_t^N|^2 f_t^N\,dZ_N + \frac{\sigma_N}{2N} \sum_{i=1}^N\int_{\R^{6N}} |\nabla_{v_i} \log \tilde f_t^{\otimes N}|^2 f_t^N \,dZ_N\\
&\quad + \frac1N \sum_{i=1}^N\|\nabla_{v_i} \log \tilde f_t^{\otimes N}\|_{L^\infty} \int_{\R^{6N}} \lt|\frac{1}{N-1}\sum_{j\neq i}k^N (x_j - x_i) - k\star\tilde\rho_{t}(x_i) \rt| f_t^N \,dZ_N\\
&\le -\frac{\sigma_N}{2N} \sum_{i=1}^N\int_{\R^{6N}} |\nabla_{v_i} \log f_t^N|^2 f_t^N\,dZ_N + \frac{\sigma_N}{2N} \sum_{i=1}^N\|\nabla_{v_i} \log \tilde f_t^{\otimes N}\|_{L^\infty}^2\\
&\quad + \frac1N \sum_{i=1}^N\|\nabla_{v_i} \log \tilde f_t^{\otimes N}\|_{L^\infty} \int_{\R^{6N}} \lt| \frac{1}{N-1}\sum_{j\neq i}k^N (x_j - x_i) - k\star\tilde\rho_{t}(x_i) \rt|f_t^N \,dZ_N\\
\end{aligned}\]
and use again the exchangeability of the particle system \eqref{micro_reg} to yield

\[\begin{aligned}
\frac{d}{dt}&\mathcal{H}_N(f_t^N | \tilde f_t^{\otimes N}) + \frac{\sigma_N}{2N}\iint_{\R^{3N}\times \R^{3N}}\sum_{i=1}^N \lt|\nabla_{v_i} \log f_t^N \rt|^2 f_t^N dZ_N \\
&\le \frac{\sigma_N}{2}\|\nabla_v \log \tilde f_t\|_{L^\infty}^2 + \|\nabla_v \log \tilde f_t\|_{L^\infty}\mathbb{E}\lt[ \bigg| k\star\tilde\rho_t (x_t^{1,N}) - \frac{1}{N-1}\sum_{j\neq 1} k^N(x_t^{1,N} -x_t^{j,N})\bigg| \rt]\\
&=: \frac{\sigma_N}{2}\|\nabla_v \log \tilde f_t\|_{L^\infty}^2 + \|\nabla_v \log \tilde f_t\|_{L^\infty}\mathbb{K}.
\end{aligned}\]
where the expectation is taken over the law $f_t^N$. Now we split the estimates for $\mathbb{K}$ as follows:
\[\begin{aligned}
\mathbb{K} &\le 5\mathbb{E}\lt[ \lt| k\star\tilde\rho_t (x_t^{1,N}) - k\star\tilde\rho_t(\bar x_t^{1,N})\rt|\rt] +5\mathbb{E}\lt[ \lt| k\star\tilde\rho_t (\bar x_t^{1,N}) -k\star\bar\rho_t (\bar x_t^{1,N}) \rt|\rt]\\
&\quad + 5\mathbb{E}\lt[\lt|(k-k^N)\star\bar\rho_t (\bar x_t^{1,N}) \rt| \rt] + 5\mathbb{E}\lt[\bigg| k^N\star\bar\rho(\bar x_t^{1,N}) - \frac{1}{N-1}\sum_{j\neq 1}k^N(\bar x_t^{1,N} - \bar x_t^{j,N})\bigg| \rt]\\
&\quad +5\mathbb{E}\lt[\bigg| \frac{1}{N-1}\sum_{j\neq 1}k^N(\bar x_t^{1,N} - \bar x_t^{j,N})- \frac{1}{N-1}\sum_{j\neq 1}k^N( x_t^{1,N} -  x_t^{j,N}) \bigg|\rt]\\
&=: \sum_{i=1}^5 \mathbb{K}_i.
\end{aligned}\]
\noindent For $\mathbb{K}_1$ and $\mathbb{K}_3$-$\mathbb{K}_5$, the estimates for $\mathbb{M}_1$ and $\mathbb{M}_3$-$\mathbb{M}_5$ in the proof of Theorem \ref{main_thm1} imply
\[\begin{aligned}
&\mathbb{K}_1 \le CN^{-\delta}\sqrt{\log N},\qquad \mathbb{K}_3 \le CN^{-\delta},\\
&\mathbb{K}_4 \le CN^{-\delta}, \qquad
\mathbb{K}_5 \le N^{-\delta}(\log N).
\end{aligned}\]
\noindent For $\mathbb{K}_2$, we let $\tilde y_t$ and  $\bar y_t$ be independent copies of $\tilde x_t$ in \eqref{char_VP} and $\bar x_t$ in \eqref{MV_reg}, respectively. Then we have
\[\begin{aligned}
\mathbb{E}&[|k\star(\tilde\rho_t - \bar\rho_t)(\bar x_t^{1,N})|]\\
&= \mathbb{E}\Big[\mathbb{E}_{(\tilde y_t, \bar y_t)}\lt[|k(\tilde y_t - \bar x_t^{1,N}) -k(\bar y_t - \bar x_t^{1,N})|\rt]\Big]\\
&\le \iint_{|\tilde y - \bar y|\le N^{-1}} \lt(\intr |k(\tilde y - \bar x) - k(\bar y - \bar x)|\bar\rho_{t}(\bar x)\,d\bar x \rt) \tilde\rho_{t}(\tilde y) \bar\rho_{t}(\bar y)\,d\tilde y d\bar y\\
&\quad +\iint_{|\tilde y - \bar y|> N^{-1}} \lt(\intr |k(\tilde y - \bar x) - k(\bar y - \bar x)|\bar \rho_{t}(\bar x)\,d\bar x \rt) \tilde\rho_{t}(\tilde y) \bar \rho_{t}(\bar y)\,d\tilde y d\bar y\\
&\le CN^{-1}\log N \|\bar\rho_{t}\|_{L^1\cap L^\infty} + C\sqrt{\log N} \|\bar\rho_{t}\|_{L^1\cap L^\infty} \mathbb{E}\lt[  \sqrt{\log N}|\bar x_t - \tilde x_t|\rt]\\
&\le CN^{-1}(\log N) \|\bar\rho_{t}\|_{L^1\cap L^\infty} + C\sqrt{\log N}\|\bar\rho_{t}\|_{L^1\cap L^\infty}\lt(\mathbb{E}\lt[\sqrt{\log N} |\bar x_t - x_t| \rt] + \mathbb{E}\lt[\sqrt{\log N} |x_t - \tilde x_t| \rt] \rt)\\
&\le CN^{-1}(\log N) + C\lt(N^{-\delta}  + \sqrt{\sigma_N} + N^{-1}(\log N) \rt)\exp\lt(C\sqrt{\log N} \rt)
\end{aligned}\]
where $x_t$ satisfies \eqref{MV_lim} and we used Proposition \ref{prop_diff} and Lemma \ref{lem_diff2}.
Thus, we get
\[
\mathbb{K}_2 \le C\lt(N^{-\delta}  + \sqrt{\sigma_N} + N^{-1}(\log N) \rt)\exp\lt(C\sqrt{\log N} \rt)
\]
for sufficiently large $N$.

\vskip3mm
Therefore, we gather all the estimates now to yield
\[\begin{aligned}
\frac{d}{dt}&\mathcal{H}_N(f_t^N | \tilde f_t^{\otimes N}) + \frac{\sigma_N}{2N}\iint_{\R^{3N}\times \R^{3N}}\sum_{i=1}^N \lt|\nabla_{v_i} \log f_t^N \rt|^2 f_t^N dZ_N\\
&\le C\sigma_N + C\lt(N^{-\delta}+ \sqrt{\sigma_N} \rt)\exp\lt( C\sqrt{\log N}\rt)+ CN^{-\delta}(\log N)
\end{aligned}\]
and this completes the proof.

\subsection{Remarks on $L^1$-convergence}
Here, we mention that if one tries to focus only on the propagation of chaos results in $L^1$, we can simply compute relative entropy between two equations \eqref{main_eq} and \eqref{VPeq}.  Namely, let $f_t$ and $\tilde f_t$ be the solution to \eqref{main_eq} and \eqref{VPeq}, respectively, with the  same initial data $f_0$. Here, we estimate
\[\begin{aligned}
\frac{d}{dt}\mathcal{H}(f_t | \tilde f_t) &= -\sigma \intrr |\nabla_v \log f_t|^2 f_t \,dz + \intrr k\star (\rho_t - \tilde \rho_t)  f_t \nabla_v \log \tilde f_t\,dz\\
&\quad + \sigma \intrr f_t\nabla_v \log f_t \nabla_v \log \tilde f_t\,dz\\
&\le  -\frac\sigma2 \intrr |\nabla_v \log f_t|^2 f_t\,dz + \frac\sigma2 \intrr |\nabla_v \log \tilde f_t|^2 f_t\,dz\\
&\quad + \|\nabla_v \log \tilde f_t\|_{L^\infty} \intr \rho_t |k\star (\rho_t - \tilde\rho_t)|\,dx\\
&\le -\frac\sigma2 \intrr |\nabla_v \log f_t|^2 f_t\,dz + \frac\sigma2 \intrr |\nabla_v \log \tilde f_t|^2 f_t\,dz\\
&\quad + \|\nabla_v \log \tilde f_t\|_{L^\infty} \mathbb{E}[|k\star(\rho_t - \tilde\rho_t)(x_t)|],
\end{aligned}\]
and similarly to the estimates for $\mathbb{K}_2$ in the proof of Theorem \ref{cor_VP}, we can get

\[\begin{aligned}
&\mathbb{E}[|k\star(\rho_t - \tilde\rho_t)(x_t)|]\\
=& \mathbb{E}\Big[\mathbb{E}_{(y_t, \tilde y_t)}\lt[|k(y_t - x_t) -k(\tilde y_t - x_t)|\rt]\Big]\\
\le& \iint_{|y - \tilde y|\le N^{-1}} \lt|\intr |k(\tilde y - x) - k(y - x)|\rho_{t}(x)\,dx \rt| \tilde\rho_{t}(\tilde y) \rho_{t}(y)\,d\tilde y dy\\
&\quad {+}\iint_{|y - \tilde y|> N^{-1}} \lt|\intr |k(\tilde y - x) - k(y - x)|\rho_{t}(x)\,dx \rt| \tilde\rho_{t}(\tilde y) \rho_{t}(y)\,d\tilde y dy\\
\le& CN^{-1}\log N \|\rho_{t}\|_{L^1\cap L^\infty} + C\sqrt{\log N} \|\rho_{t}\|_{L^1\cap L^\infty} \mathbb{E}\lt[  \sqrt{\log N}|x_t - \tilde x_t|\rt]\\
\le& CN^{-1}(\log N) \|\rho_{t}\|_{L^1\cap L^\infty} + C\sqrt{\log N}(\sqrt{\sigma} + N^{-1}(\log N))e^{C\sqrt{\log N}}\|\rho_{t}\|_{L^1\cap L^\infty}.
\end{aligned}\]
Thus, we can get
\[
\mathcal{H}(f_t | \tilde f_t) \le C\sigma_N + C\lt(N^{-1}(\log N)+ \sqrt{\sigma_N} \rt)\exp\lt( C\sqrt{\log N}\rt).
\]
Therefore, we combine this with Theorem \ref{main_thm1} to yield

\[\begin{aligned}
\|f_t^N - \tilde f_t^{\otimes N}\|_{L^1}^2 &\le 2\|f_t^N - f_t^{\otimes N}\|_{L^1}^2 + 2 \|f_t^{\otimes N} - \tilde f_t^{\otimes N}\|_{L^1}^2\\
&\le C\mathcal{H}_N\lt(f_t^N \ | \ f_t^{\otimes N} \rt) + C\mathcal{H}_1(f_t \ | \ \tilde f_t)\\
&\le C\frac{\exp(C\sqrt{\log N})}{\sigma_N N^{2\delta}} + C\sigma_N + C\lt(N^{-1}(\log N)+ \sqrt{\sigma_N} \rt)\exp\lt( C\sqrt{\log N}\rt).
\end{aligned}\]
Thus, by choosing $\sigma_N := o(N^{-2\delta}\exp(C\sqrt{\log N})$, we have the desired propagation of chaos in $L^1$. Here, we also note that if $\sigma_N \gg N^{-2\delta} \exp(C\sqrt{\log N})$, we may expect that the $N$-particle distribution $f_t^N$ tends to the solution $\tilde f_t$ of Vlasov-Poisson without getting closer to Vlasov-Poisson-Fokker-Planck.

\end{document}